%% file: conference.tex
\documentclass[10pt, conference, letterpaper]{IEEEtran}
\IEEEoverridecommandlockouts

\usepackage{cite}
\usepackage{amsmath,amssymb,amsfonts}
\usepackage{graphicx}
\usepackage{textcomp}
\usepackage{xcolor}

\usepackage{algorithmic}
\usepackage{algorithm}

\usepackage{array}
\usepackage{pifont}
\usepackage{amsthm}

\usepackage{url}

\usepackage{amsfonts}
\usepackage{booktabs}
\usepackage{siunitx}

\usepackage{subfigure}

\newtheorem{theorem}{Theorem}
\newtheorem{definition}{Definition}

\def\BibTeX{{\rm B\kern-.05em{\sc i\kern-.025em b}\kern-.08em
    T\kern-.1667em\lower.7ex\hbox{E}\kern-.125emX}}

\begin{document}

\title{DPBalance: Efficient and Fair Privacy Budget Scheduling for Federated Learning as a Service\\

\thanks{This work is supported by the National Natural Science Foundation of China (Grant No. 62302292), the National Key R\&D Program of China (Grant No. 2023YFB2704400), and the Fundamental Research Funds for the Central Universities.
The work of C. Chen is supported by National Natural Science Foundation of China (Grant No. 62202300) and Shanghai Pujiang Program (Grant No. 22PJ1404600).
}
}

\author{\IEEEauthorblockN{Yu Liu\IEEEauthorrefmark{1},
Zibo Wang\IEEEauthorrefmark{1},
Yifei Zhu\IEEEauthorrefmark{1}\IEEEauthorrefmark{3}, and Chen Chen\IEEEauthorrefmark{4}}
\IEEEauthorblockA{\IEEEauthorrefmark{1}UM-SJTU Joint Institute, Shanghai Jiao Tong University}
\IEEEauthorblockA{\IEEEauthorrefmark{3}Cooperative Medianet Innovation Center (CMIC), Shanghai Jiao Tong University}

\IEEEauthorblockA{\IEEEauthorrefmark{4}School of Electronic Information and Electrical Engineering, Shanghai Jiao Tong University}
\IEEEauthorblockA{
Email: yu\_liu@sjtu.edu.cn, wangzibo@sjtu.edu.cn, yifei.zhu@sjtu.edu.cn, chen-chen@sjtu.edu.cn}
}
\maketitle

\begin{abstract}
Federated learning (FL) has emerged as a prevalent distributed machine learning scheme that enables collaborative model training without aggregating raw data. Cloud service providers further embrace Federated Learning as a Service (FLaaS), allowing data analysts to execute their FL training pipelines over differentially-protected data. Due to the intrinsic properties of differential privacy, the enforced privacy level on data blocks can be viewed as a privacy budget that requires careful scheduling to cater to diverse training pipelines. Existing privacy budget scheduling studies prioritize either efficiency or fairness individually. In this paper, we propose DPBalance, a novel privacy budget scheduling mechanism that jointly optimizes both efficiency and fairness. We first develop a comprehensive utility function incorporating data analyst-level dominant shares and FL-specific performance metrics. A sequential allocation mechanism is then designed using the Lagrange multiplier method and effective greedy heuristics. We theoretically prove that DPBalance satisfies Pareto Efficiency, Sharing Incentive, Envy-Freeness, and Weak Strategy Proofness. We also theoretically prove the existence of a fairness-efficiency tradeoff in privacy budgeting. Extensive experiments demonstrate that DPBalance outperforms state-of-the-art solutions, achieving an average efficiency improvement of $1.44\times \sim 3.49 \times$, and an average fairness improvement of $1.37\times \sim 24.32 \times$.

\end{abstract}


\section{Introduction}
The proliferation of data on edge devices and the advancement of machine learning models drive the wide deployment of machine learning services. However, existing model training paradigm has raised severe privacy concerns since sensitive information needs to be uploaded to central servers to be analyzed. In response to this, privacy protection laws like the General Data Protection Regulation of EU (GDPR) \cite{GDPR_2018EU} have been established worldwide, which further drives the emergence of federated learning (FL). FL is a privacy-preserving distributed machine learning scheme where edge devices collaboratively train a global model while keeping their data local. It has been widely studied and deployed nowadays in applications such as Google keyboard \cite{chen2019federated, ramaswamy2019federated, tim2018applied}, hotword detection \cite{leroy2019federated}, and medical research \cite{de2019federated}.

With the advent of FL techniques, there has been a notable rise in distributed machine learning platforms like FedML \cite{he2020fedml}, FATE \cite{FATE2019} that offer a cloud-based ``federated learning as a service'' (FLaaS) \cite{FLaaS2020Nicolas}. With the lifetime of a federated model training process being represented as a pipeline, these platforms facilitate training of diverse pipelines from data analysts using crowdsourced data owners. 
To further protect and manage privacy loss, differential privacy (DP) is adopted in the FLaaS platform \cite{abadi2016deep}\cite{wei2020federated}. 
It ensures the output of FL pipeline would not change much for any neighbor input datasets. DP is enforced by adding random noise to intermediate data like gradient \cite{Sun21LDP} of FL training. With the help of DP mechanisms, data owners can specify a privacy level on
their continuously generated data blocks. Each data analyst submits multiple FL pipelines with different training demands on the amount of data and its privacy level.

According to the definition of DP, the training of a FL model on a data block introduces a certain level of privacy loss. Furthermore, this privacy loss is additive when data is utilized for multiple times. As a result, privacy in FLaaS platforms can be perceived as a consumable and quantifiable resource. The privacy level set by each crowdsourced data owner can be regarded as a privacy resource budget. Consequently, the efficient scheduling of FL pipelines to train on crowdsourced private data within the privacy budget becomes a critical challenge for FLaaS platforms.

However, while privacy can be treated as a resource, traditional resource allocation methods in distributed computing cannot be directly applied.
The primary distinction lies in the nature of privacy resources, which are characterized as \textit{non-replenishable} and \textit{continuously generated} as new data is collected.  Unlike traditional resources such as CPU and memory which can be released for subsequent tasks after prior usage, privacy resources are not recoverable once their budget has been depleted. Furthermore, allocating available resources solely based on demand, a strategy that may work for CPU/ memory allocation, may be however inefficient or unfair for privacy resources.

Existing works in privacy budget scheduling solely consider either efficiency \cite{ghodsi2011dominant, li2022dplanner, tholoniat2022packing, yuan2022Privacyas, kuchler2023cohere} or fairness \cite{luo2021privacy}. So far there is no privacy budget scheduling mechanism that jointly considers them both. In fact, efficiency and fairness of an allocation scheme are complicatedly affected by many factors.
The FL platform may hope to let the pipeline await to obtain a boarder view of the allocation market to achieve a better performance in system throughput and fairness. However, requiring the pipelines to wait for a long time increases the latency, which hurts the final efficiency from a different perspective.
In addition, naively allocating the majority of resources to those low-demand pipelines to maximize the number of completed pipelines may make the allocation among data analysts imbalanced, which hurts fairness.
Therefore, it is challenging to derive an allocation comprehensively taking efficiency and fairness into consideration in the context of FLaaS platform.

In this paper, we present the first privacy budget scheduling mechanism, DPBalance, that jointly considers efficiency and fairness. 
We first develop a comprehensive utility function incorporating dominant shares across different FL pipelines to calculate data analyst-level fairness and platform-level efficiency. To capture the characteristics of FLaaS training, FL-specific performance metrics, such as the matching degree of data and FL models, are further integrated into the utility model. A sequential allocation mechanism is then designed using the Lagrange multiplier method and effective greedy heuristics. 
From the perspective of fairness, 
we theoretically prove that our solution can simultaneously satisfy four key economic properties. From the perspective of efficiency, the most efficient allocation result under the privacy preference is derived by maximizing platform's utility function. 
In summary, the contributions of our work are as follows. 

\begin{itemize}
\item We present the first privacy budget scheduling mechanism, DPBalance, that jointly considers efficiency and fairness in the context of FLaaS model training.
\item We design a sequential allocation algorithm using the Lagrange multiplier method and effective greedy heuristics so that data analyst-level fairness and platform-level efficiency can be maximized. 

\item We theoretically prove that DPBalance satisfies four essential economic properties: Pareto Efficiency, Sharing Incentive, Envy-Freeness, and Weak Strategy Proofness.
\item We discover and theoretically prove the existence of a tradeoff between fairness and efficiency under practical conditions.
\item Extensive experiments show that DPBalance outperforms other state-of-the-art budget scheduling baselines by $1.44 \times \sim 3.49 \times$ in efficiency and $1.37 \times \sim 24.32 \times$ in fairness on average.

\end{itemize}

The remainder of this paper is constructed as follows: Section \ref{Related work} summarizes the related work. Background and motivation are presented in Section \ref{Background and motivation}. In Section \ref{System model and problem formulation}, we present the system model and problem formulation of privacy budget scheduling. Algorithm design and theoretical analysis are presented in Section \ref{Algorithm design and theoretical analysis}. The evaluation part is presented in Section \ref{Evaluation}, followed by the conclusion in Section \ref{Conclusion}.

\section{Related work}
\label{Related work}

\subsection{Resource allocation in FL}
Existing resource allocation studies in FL mainly focus on allocating traditional resources, like energy consumption \cite{dinh2020federated}, bandwidth \cite{shi2020joint} or computation resources \cite{li2019fair, lim2021decentralized} to train a single model. In FEDL \cite{dinh2020federated}, the uplink transmission rate is allocated to minimize energy cost and training time. In \cite{shi2020joint}, a joint device scheduling and resource allocation policy is proposed to maximize model accuracy with a limited training time budget and latency threshold. 
In q-FEL \cite{li2019fair}, fairness in accuracy across different users is considered when allocating edge devices. 
In \cite{lim2021decentralized}, a deep learning-based auction mechanism is proposed to infer the valuation of each cluster head's service to allocate clusters' computation resources. However, none of these works consider privacy as resources and study privacy allocation problems in the emerging FLaaS setting, where multiple FL models need to be trained on top of privacy-preserved data owners. 

\subsection{Privacy budgeting for distributed model training}

Current studies on privacy budget scheduling can mainly be categorized into two types: efficiency-oriented \cite{ghodsi2011dominant, li2022dplanner,  tholoniat2022packing, yuan2022Privacyas, kuchler2023cohere } and fairness-oriented approaches \cite{luo2021privacy}.
In \cite{li2022dplanner}, privacy budget is allocated to maximize the overall quality of query results with an online budgeting algorithm. In \cite{tholoniat2022packing, yuan2022Privacyas, kuchler2023cohere}, privacy budget allocation problem is formulated to maximize the number of allocated pipelines, and adopt heuristic algorithms \cite{tholoniat2022packing, yuan2022Privacyas}, or solver \cite{kuchler2023cohere} to resolve it. 
DPF \cite{luo2021privacy} focuses on the fairness in allocation and is designed to ensure first $N$ pipelines' max-min fairness at the pipeline level.
Different from these studies, we propose the first work to consider both fairness and efficiency in privacy budget scheduling.  We further present a comprehensive comparison of our work to others in Fig. \ref{Comparison of privacy scheduling}.

\begin{figure}[t]
\centering 
\includegraphics[width=\linewidth]{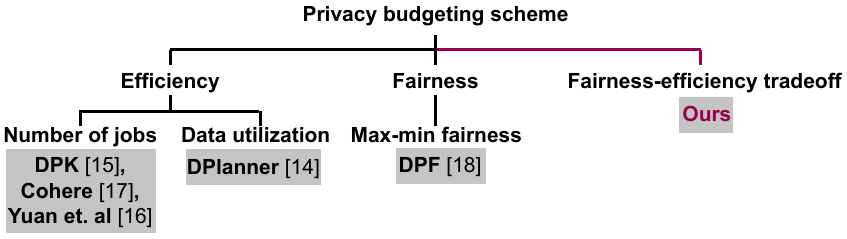}
\caption{Taxonomy of existing privacy budgeting studies and ours (DPBalance).}
\label{Comparison of privacy scheduling}
\end{figure}

\section{Background and motivation}
\label{Background and motivation}

\subsection{Why privacy can be allocated?}

In this subsection, we introduce bounded property, composition of privacy to illustrate why privacy can be regarded as resources to be allocated.

To show the bounded property, we introduce the state-of-the-art DP, R\'{e}nyi Different Privacy \cite{RENYI}\cite{mironov2017renyi}, and the resulting privacy loss from it. We first give the formal definition of R\'{e}nyi divergence:

\begin{definition}[R\'{e}nyi Divergence]
\label{Renyi divergence} 
Let $\Gamma$ and $\Upsilon$ be two distributions, R\'{e}nyi divergence $\mathcal{RD}_{\alpha}$ of order $\alpha$ is defined as

\begin{equation}
    \mathcal{RD}_{\alpha} (\Gamma || \Upsilon) \triangleq \frac{1}{\alpha-1} \log E_{x \sim \Upsilon}(\frac{\Gamma(x)}{\Upsilon(x)}),
\end{equation}
where $\Gamma(x)$ and $\Upsilon(x)$ are densities of $\Gamma$ and $\Upsilon$ at point $x$.

\end{definition}
 
Then, R\'{e}nyi Different Privacy is defined as follows:

\begin{definition}[($\alpha, \epsilon$)-R\'{e}nyi Different Privacy (RDP)]

Let $M$ be the random mechanism processing on two adjacent datasets $D$ and $D'$ that only differ in one record. ($\alpha, \epsilon$)-R\'{e}nyi Different Privacy holds that

\begin{equation}
\label{eq3}
\begin{aligned}
    & \mathcal{RD}_{\alpha} (M(D)||M(D')) \le \epsilon,
\end{aligned}
\end{equation}
where R\'{e}nyi divergence $\mathcal{RD}_{\alpha} (M(D)||M(D'))$ is equivalent to the privacy loss of dataset $D$.
\end{definition}

We call the $\epsilon$-bounded privacy loss as the bounded property. When $\alpha$ approaches to $\infty$, ($\alpha, \epsilon$)-RDP is equivalent to the basic $(\epsilon, 0)$-DP. RDP is widely applied in the privacy budgeting solutions \cite{luo2021privacy, tholoniat2022packing, yuan2022Privacyas, kuchler2023cohere} for its advantages over basic DP: 1) it bounds privacy loss tighter when achieving similar performance in FL; 2) it permits more convenient composition of multiple different DP mechanisms including Gaussian mechanism and Laplacian mechanism, and scales better. 

As different data analyst requires different amount of data, we partition growing dataset of each FL device into data blocks based on time, namely $\mathcal{S} = \{d_k| k\in[1,\infty)\}$, where $\mathcal{S}$ is the data block set. In this way, each data analyst can specify different number of data blocks.
To make each data block reusable and accountable for privacy loss each time being used, it's essential to introduce parallel composition and sequential composition, which depict the additive property of privacy.

\begin{definition}[Parallel composition]
We denote each data block $i$'s privacy loss as $\epsilon_i$. Then, the privacy loss for the whole growing dataset is $\max \epsilon_i$, $\forall i \in [1, \infty)$.
\end{definition}

\begin{definition}[Sequential composition]
Let the total privacy budget of FL device be $\epsilon_g$. We assume it has been used $K$ times by different FL pipelines, and each time it is consumed $\epsilon_i$ privacy budget. Then, the privacy loss for this data block can be $\sum_{i=1}^K \epsilon_i$.
\end{definition}

Intuitively, parallel composition guarantees the privacy loss of dataset is bounded by the maximum privacy loss of its single privacy loss. Sequential composition depicts that privacy loss is additive for a single data block. It is straightforward to prove that these two properties hold perfectly for $(\alpha, \epsilon)$-RDP.

As privacy is limited due to the bounded property and additivity due to the composition property, it can be regarded as a special resource to be allocated.

\subsection{Effects of bad scheduling}

As mentioned, privacy can be considered as an important resource to be allocated. When the limited resource is allocated to multiple pipelines, efficiency and fairness of the allocation becomes a critical metric \cite{ghodsi2011dominant, li2019fair, uchida2009information, lan2010axiomatic, joe2013multiresource}.

Consider that data blocks are constantly contributed by edge devices and data analysts constantly coming in with their pipelines. Each pipeline demands certain data blocks with privacy demands. Privacy demand can be regarded as the amount of privacy loss of a certain data block specified by a data analyst. In our context, dominant share is defined as the maximum privacy budget among all required data blocks of a training demand from a pipeline or data analyst. Training demand includes number of data blocks and privacy demand on each data block. 

\begin{definition} [Pipeline's maximum share]

A pipeline's maximum share of all the required data block's privacy resources is defined as

\begin{equation}
    \mu_{ij} = \max_k \gamma_{ij}^{<k>},
    \label{relative dominant share}
\end{equation}
where $\gamma_{ij}^{<k>}$ is the normalized privacy demand of pipeline $j$ from data analyst $i$ for data block $k$.
\end{definition}

Thus, this pipeline's dominant share can be represented as $\mu_{ij} x_{ij}$. $x_{ij}$ is the ratio of allocated resource to privacy demand of pipeline $j$ from data analyst $i$, which can also be considered as allocated result to pipeline $j$.

\begin{definition} [Data analyst's maximum share]
Data analyst $i$'s maximum share of all the required data block's privacy resource is defined as

\begin{equation}
    \mu_{i} = \max_k \gamma_{i}^{<k>}, 
    \label{data analyst's relative dominant share}
\end{equation}
where $\gamma_{i}^{<k>}$ is the normalized privacy demand that data analyst $i$ wants for data block $k$ on all selected edge devices, namely, $\gamma_{i}^{<k>} = \sum_{j=0}^{n_i} \gamma_{ij}^{<k>} x_{ij}$.

\end{definition}

Thus, data analyst $i$'s dominant share is represented as $\mu_{i} x_{i}$, which indicates the degree to its demand was satisfied. $x_i$ is the ratio of allocated resource to privacy demand of data analyst $i$, which can be considered as allocation to data analyst $i$.

\begin{figure}[t]
\centering 
\includegraphics[width=9cm]{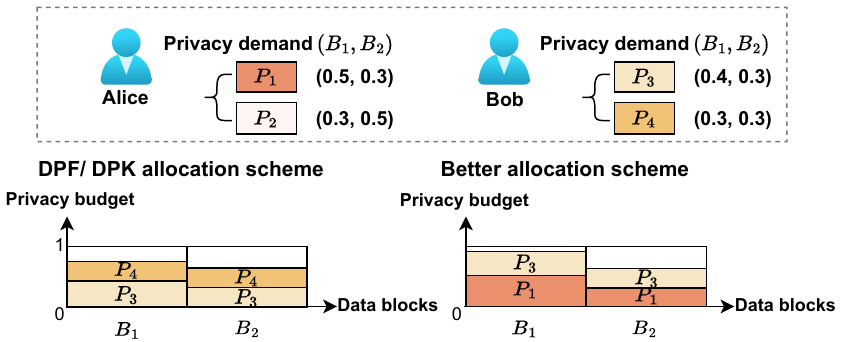}
\caption{Allocation results under different schemes for a particular time slot.}
\label{Allocation example}
\end{figure}

To illustrate the shortcomings of bad scheduling, we give an example in Fig. \ref{Allocation example}. Consider two data analysts named Alice and Bob, two data blocks $B_1, B_2$ whose total privacy budget is $1.0$. Alice has two pipelines ($P_1, P_2$). Pipeline $P_1$ demands ($0.5, 0.3$) for ($B_1, B_2$), where the normalized privacy demand ($\gamma^{<1>}, \gamma^{<2>}$) for ($B_1, B_2$) is ($0.5, 0.3$) and maximum share $\mu$ is $0.5$. 
Pipeline $P_2$ demands ($0.3, 0.5$) for ($B_1, B_2$), where the normalized privacy demand ($\gamma^{<1>}, \gamma^{<2>}$) for ($B_1, B_2$) is ($0.3, 0.5$) and maximum share $\mu$ is $0.5$.  
Bob also has two pipelines $(P_3, P_4)$.
Pipeline $P_3$ demands ($0.4, 0.3$) for ($B_1, B_2$), where the normalized privacy demand ($\gamma^{<1>}, \gamma^{<2>}$) for ($B_1, B_2$) is ($0.4, 0.3$) and maximum share $\mu$ is $0.4$.  
Pipeline $P_4$ demands ($0.3, 0.3$) for ($B_1, B_2$), where the normalized privacy demand ($\gamma^{<1>}, \gamma^{<2>}$) for ($B_1, B_2$) is ($0.3, 0.3$) and maximum share $\mu$ is $0.3$. 

In this example, DPK \cite{tholoniat2022packing}, which is the state-of-the-art privacy scheduling method focusing on efficiency and allocating resources to the pipeline with the lowest weight-to-demand ratio, allocates pipeline $P_3$ and $P_4$ shown in the left side of Fig. \ref{Allocation example}. In addition, DPF \cite{luo2021privacy}, which is the state-of-the-art privacy scheduling method focusing on fairness and allocating resources to the pipeline with the smallest dominant share, gets the same allocation result as DPK. Efficiency can be evaluated with dominant share and leftover privacy resources. Based on the definition in \eqref{relative dominant share}, the overall dominant share of DPF and DPK solution is $0.7$. From the perspective of leftover privacy resources, DPF and DPK waste $0.3\ B_1$ and $0.4\ B_2$. We can figure out another solution shown on the right side of Fig. \ref{Allocation example}, pipeline $P_1$ from Alice and pipeline $P_3$ from Bob can be satisfied at the same time. Furthermore, this allocation result leaves less unused resource, as it leaves $0.1\ B_1$ and $0.4\ B_2$ unused, and achieves $0.9$ overall dominant share. Furthermore, it's much fairer because there is no analyst who occupies all the privacy resources.

\section{System model and problem formulation}
\label{System model and problem formulation}

In this section, we first introduce the threat model and the entities involved in our framework. Then, we formally formulate the privacy resource allocation in FLaaS platforms as an optimization problem.

\subsection{Threat model}
In the FLaaS platform, the data analysts submit multiple pipelines and receive them after being trained by the private data held by vast data contributors. The trained pipelines exposed to data analysts or other entities, although does not explicitly contain the raw data, still leak sensitive information of the individual data contributors and are vulnerable to privacy threats such as membership inference attacks \cite{shokri2017membership} or data reconstruction attacks \cite{carlini2019secret}. To tackle the privacy threats, we apply the DP model, which restricts information leakage by introducing randomness to the outputs, to the trained pipelines \cite{wei2020federated}. In the FLaaS platform, one data block may be utilized in training multiple pipelines, where an adversary has an increased ability in performing attacks with an aggregated view of multiple pipelines. Therefore, we propose privacy budget scheduling mechanism in the FLaaS platform, where each data block satisfies $(\alpha,\epsilon)$-RDP over the whole execution after careful DP assignment on each pipeline.

\subsection{System overview}

\begin{figure}[t]
\centering 
\includegraphics[width=\linewidth]{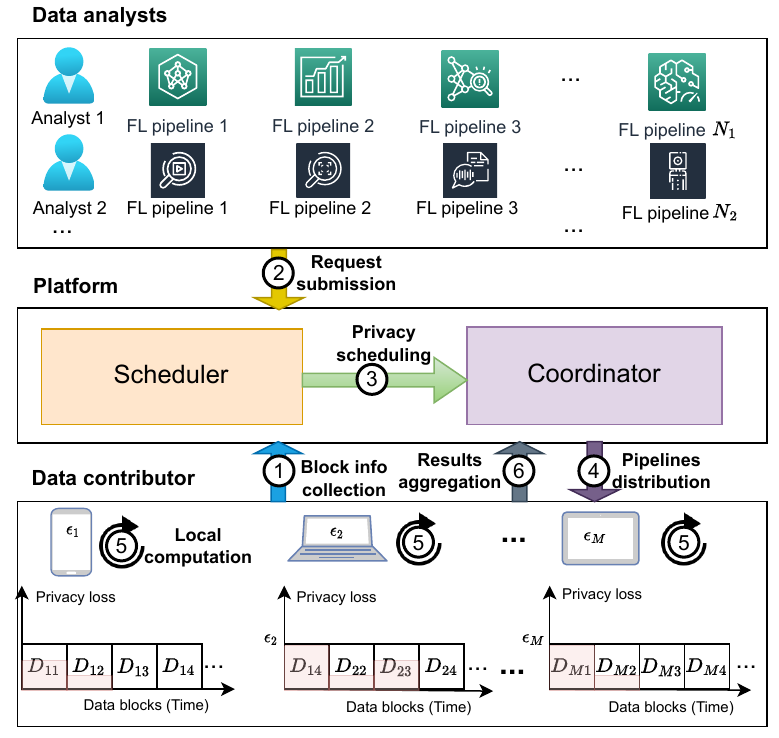}
\caption{Framework of DPBalance}
\label{framework}
\end{figure}

The framework of DPBalance is depicted in Fig. \ref{framework} and it considers three entities: Data analysts, FLaaS platform, and data contributors. Data contributors can be mobile phones, tablets, or other edge devices with growing databases. In this system, $Q$ edge devices join the FLaaS platform training and each participant has its own database partitioned in many data blocks by time. We denote the number of data blocks from all edge devices in total as $K$. We assume there is no edge device in and out and data on devices are continuously generated.

Our framework includes $M$ data analysts, and each data analyst $i$ consists of $N_i$ pipelines. FL pipelines are associated with their training demands. A pipeline can only be successfully executed if its training demand is met. Pipelines are willing to get data blocks with more privacy loss as possible because it could improve the performance of trained pipelines. In addition, we assume data analysts are greedy and selfish. Greedy means the data analyst wants to satisfy more pipelines and selfish means the data analyst wouldn't return the resource back if its pipelines could benefit from it.
The main objective of the FLaaS platform is to allocate a certain amount of data blocks efficiently and fairly under limited privacy resources. An unfair allocation scheme would discourage data analysts from participating and an inefficient allocation scheme would cause a waste of privacy resources. 

Next, we briefly introduce the workflow of our system in Fig. \ref{framework}. In the training demand formulation part, the scheduler collects data analysts' general demand, including the number of demanding devices $N$, privacy budget $\overrightarrow{\epsilon_j}$ and time interval $\overrightarrow{T_j}$ of demanding blocks, and randomly selects a set of $N$ edge devices $\overrightarrow{N_j}$. After getting the privacy demands, the scheduler allocates privacy resources and sends results to the platform. Then, the platform distributes successfully allocated pipelines to the target devices for training and aggregates.

\subsection{Privacy resource model}
We partition each device's local data into many data blocks by time and partition granularity can be hour, day, or month. Each FL pipeline can specify the data blocks from a certain range of time. Every time the FL pipeline is trained on a local device, data blocks would lose part of their privacy. 
As time goes by, private data used for training FL pipelines would grow and get retired after exposing a certain amount of privacy. Data blocks can serve other FL pipelines before retirement. 

We define the privacy budget of data contributor $i$'s whole dataset as $\epsilon_i^g$ and any privacy budget of data block $j$ of it as $\epsilon_{ij}^g$. We define device $i$'s privacy loss as $\epsilon_i^c$ and any privacy budget of data block $j$ of it as $\epsilon_{ij}^c$. We assume privacy loss should not exceed privacy budget ($\epsilon_i^c \le \epsilon_i^g, \epsilon_{ij}^c \le \epsilon_{ij}^g$). Furthermore, we assume different edge device $i$ has different privacy budget $\epsilon_i^g$. In addition, to keep $(\alpha, \epsilon_i)$-RDP for edge device $i$, data block $j$ on edge device $i$ holds the same privacy budget as the global budget held by each edge device ($\epsilon_{ij}^g=\epsilon_i^g$). Furthermore, we assume the privacy allocation follows one-or-more property. 

\begin{definition}[One-or-more]
A pipeline from a data analyst would be successfully allocated if and only if its allocated resource is equal to or more than its minimal demand. Formally speaking, it can be represented as follows:

\begin{equation}
    \mathcal{D'} = \varkappa \mathcal{D}, \varkappa \ge 1,
\end{equation}
where $\mathcal{D}$ is the privacy demand vector of this pipeline, and $\mathcal{D'}$ is the vector representing the amount of privacy resource allocated to the pipeline.

\end{definition}

This property is unique to privacy resources, while traditional resources (e.g. CPU, memory) do not have. 

\subsection{Data analyst's efficiency model}
Before defining the platform's utility, we define the data analyst's efficiency first. We adopt dominant share, weighted by the data-pipeline matching degree and pipeline delay time as the efficiency of a data analyst. Pipeline delay time means how long a pipeline has been waiting since it proposes training demand. The data-pipeline matching degree means the importance of a certain data block to a certain pipeline measured by training loss \cite{li2022dplanner} \cite{lai2021oort}. We assume longer delay time and lower training loss bring lower efficiency. We use the weighted average of the training losses $l_{ij}$ of all data blocks on device $i$ to represent the training loss $l_i$ of this device.

\begin{equation}
    l_i = \sum_{j=1}^n \frac{\mu_{ij}}{\sum_{j=1}^n \mu_{ij}} l_{ij}.
\end{equation}

The data analyst's efficiency is defined as follows:

\begin{definition} [Data analyst's efficiency] 
Data Analyst's Efficiency is defined as data analyst $i$'s dominant share weighted by waiting time coefficient $T(t_i)$ and training loss $l_i$.
\begin{equation}
    U_i (x_i) = \mu_i x_i T(t_i) l_i,
\end{equation}
where $T(t_i)$ can be any monotonic decreasing function of data analyst's waiting time $t_i$.
\end{definition}

\subsection{Platform's utility model}
The platform's utility considers two perspectives. First, it pursues a higher overall efficiency for all data analysts. Second, it tries to balance the efficiency of the data analysts for fairness. Therefore, we present to model the platform's utility incorporating both efficiency and fairness.

For the efficiency, we define the following dominant efficiency to represent the platform's efficiency.

\begin{definition}[Dominant efficiency] 
Dominant efficiency is defined as the summation of all data analysts' efficiency
\begin{equation}
    \sum_{i=1}^m U_i(x_i).
    \label{Dominant efficiency}
\end{equation}
\end{definition}

For the fairness metric, we consider data analyst-level fairness, which means we consider fairness among data analysts. We give the formal definition of dominant fairness.

\begin{definition}[Dominant fairness]
We define dominant fairness $f_{\beta}(x_i)$ from the perspective of dominant share $\mu_i x_i$, average delay time $t_i$, and weighted average training loss $l_i$.
\begin{equation}
    f_{\beta}(x_i) =sgn(1-\beta) \bigg(\sum_{i=1}^m\Big(\frac{U_i(x_i)}{\sum_{i=1}^m U_i(x_i) }\Big)^{1-\beta}\bigg)^{\frac{1}{\beta}},
    \label{Dominant fairness}
\end{equation}

where $sgn(\cdot)$ is the signum function and $\beta$ can be considered as the fairness preference of platform manager.
\end{definition}

After that, we combine the efficiency metric and fairness metric and formulate the platform's utility function as

\begin{equation}
\label{opt 1}
    \begin{aligned}
        & \Psi_{\lambda}(x_i)=sgn(1-\beta) \bigg(\sum_{i=1}^m \Big(\frac{\mu_i x_i T(t_i)l_i }{\sum_{i=1}^m \mu_i x_i T(t_i) l_i} \Big)^{1-\beta} \bigg)^{\frac{1}{\beta}} \\     
        & (\sum_{i=1}^m \mu_i x_i T(t_i) l_i)^ \lambda,
    \end{aligned}
\end{equation}
where $\lambda$ is the efficiency preference of the platform manager.

When $\lambda = |\frac{1-\beta}{\beta}|$, \eqref{opt 1} can be
transformed into 
\begin{equation}
    \max sgn(1-\beta) \bigg(\sum_{i=1}^m \Big(\mu_i x_i T(t_i) l_i \Big)^{1-\beta} \bigg)^{\frac{1}{\beta}}
    \label{simplified problem},
\end{equation}
which is equivalent to the famous $\alpha$-fairness function \cite{Generalized2008Eitan}

\begin{equation}
\label{alpha-fairness}
    \max \sum_{i=1}^m \frac{ \Big(\mu_i x_i T(t_i)l_i \Big)^{1-\beta}}{1-\beta}.
\end{equation}

\subsection{Problem formulation}
From the perspective of the FLaaS platform, we aim at maximizing the platform's utility. Formally, we can formulate our problem as

\begin{align}
 \label{global opt} \max\quad &\Psi_{\lambda}(x_i), \\ 
 \label{resource availability}
    s.t. \quad &\sum_{i=1}^m \gamma_{i}^{<k>}x_i T(t_i) l_i \le 1, \forall k \in [1, K],\\
 \label{demand composition}
   &\gamma_{i}^{<k>} = \sum_{j=0}^{n_i} \gamma_{ij}^{<k>} x_{ij}, \forall i \in [1, M], \forall k \in [1, K], \\ 
  \label{One-or-more property}
  &x_{ij} = 0 \text{ or } x_{ij} \ge 1. 
\end{align}

This optimization problem is constrained by resource availability in \eqref{resource availability}, where $\gamma_{i}^{<k>}$ specifies the normalized demand that data analyst $i$ wants for resource $k$. Constraint \eqref{demand composition} represents the demand of data analyst $i$ is composed of demands of all its pipelines. Constraint \eqref{One-or-more property} is designed because of one-or-more property.

\section{Algorithm design and theoretical analysis}
\label{Algorithm design and theoretical analysis}

\subsection{Algorithm design}

The original optimization problem is non-convex and discontinuous because constraint in (\ref{One-or-more property}) contains a discrete point $x_{ij}=0$. Classical solutions for mixed integer optimization problems, including the branch and bound method, dynamic programming, and cutting plane method, are iterative and cannot provide a closed-form solution, which makes it hard to analyze the economic properties. 
Therefore, we decompose the original problem into two sub-problems.

The first sub-problem is to allocate privacy resources among data analysts and is formulated as
\begin{align}
 \label{sub-problem1} \max\quad &\Psi_{\lambda}(x_i), \\ 
 \label{resource availability1}
    s.t. \quad &\sum_{i=1}^m \gamma_{i}^{<k>}x_i T(t_i) l_i \le 1, \forall k\in[1,K],\\
    \label{new constraint}
    & x_i \ge 0.
\end{align}

The second sub-problem is to reallocate privacy resources among pipelines in one data analyst to cover more pipelines. It is formulated as
\begin{align}
 \label{sub-problem2} \max \quad & \sum_{j=1}^n \mu_{ij} x_{ij} T(t_{ij}) l_{ij},\\ 
 \label{resource availability2}
    s.t. \quad &\gamma_{i}^{<k>} x_i = \sum_{j=0}^{n_i} \gamma_{ij}^{<k>} x_{ij}, \forall i \in [1,M], \\
    \label{new constraint1}
    & x_{ij} = 0 \text{ or } x_{ij} \ge 1.
\end{align}

We then propose a novel sequential allocation algorithm to solve them. The detailed algorithm is presented in Algorithm. \ref{sequential allocation algorithm}. After decomposition, the first optimization problem becomes convex and continuous. We first assemble pipelines' training demands to formulate data analyst $i$'s training demand and calculate its maximum share of all required resources $\mu_i$ using \eqref{relative dominant share} (line 1). We then solve the optimization problem in \eqref{sub-problem1} along using Lagrange multiplier method and output each data analyst's allocation result $X$ (line 2).

We solve the second sub-problem in a greedy heuristic way, and the general policy is to return unused resources to the resource pool (line 4). To satisfy the greedy and selfish characteristics of data analysts mentioned in Section \uppercase\expandafter{\romannumeral4} and cover more pipelines in one data analyst, we temporarily relax the one-or-more property in \eqref{new constraint1} into integer constraint shown in \eqref{all or nothing}. For each data analyst $i$, we get the successfully allocated pipelines set $\chi_i$ by solving problem in \eqref{second step} (line 5). 

\begin{equation}
    \max \sum_j \Gamma (x_{ij}), \forall i \in [0,M],
    \label{second step}
\end{equation}
where
\begin{equation}
    \Gamma (x) = 
    \quad
    \begin{cases}
        \begin{aligned}
            & 1, x_{ij} \ge 1,   \\
           
            & 0, x_{ij} = 0. \\ 
        \end{aligned}
    \end{cases} \\
    \label{all or nothing}
\end{equation}

This action is to allocate fewer resources to each analyst and select as many pipelines as possible. Secondly, to make full use of allocated resources for data analysts, we choose to set $x_{ij} \in \chi_i$ and maximize the second sub-problem in \eqref{sub-problem2} using the commercial Gurobi solver \cite{gurobi2022gurobi} (line 6).
Each data analyst $i$ returns unused privacy resource back if there is privacy resource left or its $\gamma_{ij}$ from pipeline $j$ with minimal $\mu_{ij}$ Pareto-dominates than data analyst $i$'s share $\gamma_i x_i$ because of one-or-more property (line 7), where the latter means left resource can not meet the minimal requirement of any pipeline in the data analyst.
This behavior is also a kind of economizing behavior, which can give resources to the data analyst who will come continuously in the future. Finally, the coordinator assembles all pipelines' allocation result and form the global allocation result $X'$ (line 8).

\begin{algorithm}[t]
\caption{Sequential allocation algorithm in DPBalance}
\label{sequential allocation algorithm}
\begin{algorithmic}[1] 
\REQUIRE ~~\\ 
    The training demand of all the pipelines, $d_{ij} = [d_{ij}^{<1>}, d_{ij}^{<2>},...,d_{ij}^{<K>}], i \in [1, M], j \in [1, N]$; \\
    The delay time of all the data analysts' pipelines, $t=[t_1, t_2,...,t_M]$;\\
    The training loss of all the data analysts' pipelines, $l=[l_1, l_2,...,l_M]$, where $l_i=[l_{i1}, l_{i2},..., l_{iM}]$;\\
    The hyper-parameter $\beta$ and $\lambda$;\\
\ENSURE ~~\\ 
    Allocation result, $X' = [x_1, x_2,...,x_m], x_i = [x_{i1}, x_{i2},...,x_{iN}], \forall i \in [1, M]$; \\
    
    \STATE Calculate data analyst $i$'s maximum share of all the required resource $\mu_i$ and $\gamma_{ij}$ using \eqref{relative dominant share};  
    \STATE Solve the first sub-problem in \eqref{sub-problem1} along with constraints in \eqref{resource availability1}) and \eqref{new constraint}, and output the data analyst $i$'s primary allocation result $X$;\\  
    \FOR{each data analyst $i$}
    \STATE Return allocated resource if $\gamma_{ij}$ from pipeline $j$ with lowest $\mu_{ij}$ Pareto-dominates $\gamma_i x_i$ and break;  \\
    
    \STATE Maximize the optimization problem in \eqref{second step} with constraints and get the successful allocated pipeline set $\chi_i$;\\
    
    \STATE Maximize optimization problem in \eqref{sub-problem2} and form pipeline's allocation result $x_{ij}$; \\
    
    \STATE Return unused privacy resources back. \\
    
    \STATE Assemble all the pipelines' allocation results from all data analysts and form the final allocation result $X'$; \\
    \ENDFOR

\RETURN $X'$; 
\end{algorithmic}
\end{algorithm}

We further use the allocation example in Fig. \ref{Allocation example} to illustrate how our algorithm works. Firstly, we set fairness preference $\beta=2.2$ and efficiency preference $\lambda = \frac{\beta-1}{\beta}$. Then, we assemble data analysts' training demands by adding up all its pipelines' demand. After that, we solve the optimization problem in \eqref{sub-problem1} along with its constraints and get two data analysts' output. Alice gets $(0.5, 0.5)$ for ($B_1, B_2$), and Bob gets $(0.5, 0.42)$ for ($B_1, B_2$). As both data analysts are allocated with partial resources, they do not need to return resources to the resource pool. Then, by maximizing the optimization in \eqref{sub-problem2} along with constraints, we can get the resource allocation of two data analysts to their pipelines. Alice allocates $(0.5, 0.3)$ of ($B_1, B_2$) to pipeline $P_1$. Bob allocates $(0.5, 0.375)$ of ($B_1, B_2$) to pipeline $P_3$. Finally, Alice returns $0.2 B_2$ and Bob returns $0.045 B_2$ back. In this way, our algorithm achieves $1.0$ dominant efficiency and allocates $2.25$ pipelines, while DPF and DPK only achieve $0.7$ dominant efficiency and allocates $2$ pipelines successfully.

\subsection{Four key properties of fairness}
In this subsection, we introduce four definitions of key properties for fairness, and prove under what conditions these properties can be satisfied.

\begin{definition}[Pareto Efficiency (PE)] 
Suppose that we have 2 vectors \textbf{x} and \textbf{y}, if $x_i \ge y_i$ for any index $i$ and $x_j > y_j$ for some index $j$, we denote \textbf{x} Pareto-dominates \textbf{y}. 
\end{definition}

In other words, data analyst cannot increase its utility without decreasing others' utility. Mathematically speaking, if demand $x_i'$ Pareto-dominates $x_i$, then $\Psi(x_i') > \Psi(x_i)$.

\begin{theorem}
\label{Theorem PE}
The solution for the optimization problem in \eqref{opt 1} is PE if and only if when $\beta>0$ and $\vert{\lambda} \vert \ge \vert{\frac{1-\beta}{\beta}} \vert$.
\end{theorem}

\begin{proof}
The detailed proof can be found in Appendix \ref{proof of theorem 1}.
\end{proof}

\begin{definition} [Sharing Incentive (SI)] 
A sharing incentive allocation $x_i$ to data analyst $i$ satisfies the sharing incentive property if each data analyst $i$ gains more utility than the evenly fair share result $x'_i$, i.e.,

$$U_i(x_i) \ge U_i(x'_i).$$
\end{definition}

In other words, there is no incentive for data analysts to ask servers to partition privacy resources equally.

\begin{theorem}
\label{theorem SI}
The solution to the original optimization problem in \eqref{global opt} satisfies SI (or not) as follows.

(a) SI is satisfied when $\beta>1$ and $\lambda = \frac{\beta-1}{\beta}$.

(b) SI is not satisfied when $0<\beta<1$ and $\lambda = \frac{\beta-1}{\beta}$.

(c) SI is satisfied when $\lambda = 0$ for any $\beta$.

(d) SI is not satisfied when $\lambda = \infty$ for any $\beta$.
\end{theorem}

\begin{proof}
The detailed proof can be found in Appendix \ref{proof of theorem 2}.
\end{proof}

Next, we would like to define envy-freeness:

\begin{definition} [Envy-Freeness (EF)] 
This property holds if and only if any data analyst $i$ would not envy other data analyst's allocation. Mathematically speaking, data analyst $i$ and data analyst $j$ get the allocation $x_i$ and allocation $x_j$ respectively in our allocation scheme.  Suppose that data analyst $i$ envies data analyst $j$'s allocation $x_j$, it's impossible that data analyst $i$ gets more utility with data analyst $j$'s allocation result $x_j$.
\end{definition}

\begin{theorem}
\label{Theorem EF}
The solution to the original optimization problem in \eqref{global opt} satisfies EF (or not) as follows.

(a) EF is satisfied when $\beta>1$ and $\lambda = \frac{\beta-1}{\beta}$.

(b) EF is not satisfied when $0<\beta<1$ and $\lambda = \frac{\beta-1}{\beta}$.

(c) EF is satisfied when $\lambda = 0$ for any $\beta$.

(d) EF is not satisfied when $\lambda = \infty$ for any $\beta$.
\end{theorem}

\begin{proof}
The detailed proof can be found in Appendix \ref{proof of theorem 3}.
\end{proof}

\begin{definition} [Strong Strategy Proofness (SSP)] 
It depicts that a data analyst can not increase both its utility and non-dominant share by lying about its training demands.
\end{definition}

\begin{definition} [Weak Strategy Proofness (WSP)] 
It depicts that when a data analyst lies about its training demands, it may increase its utility but decrease its non-dominant share.
\end{definition}

Both definitions of Strategy Proofness (SP) mean data analyst has no incentive to lie about training demands. Under SSP, there would be no benefit from lying. Under WSP, lying about training demand may increase its utility, but it may also reduce non-dominant shares. Unlike the traditional resource allocation where training demands must be strictly proportional to each resource, getting extra non-dominant privacy resources may also help data analyst to finish more pipelines or increase some pipeline's accuracy as the data analyst's training demand is composed of pipelines' training demands.

\begin{theorem}
\label{Theorem SP}
The solution to the original optimization problem in \eqref{global opt} satisfies SP (or not) as follows.

(a) WSP is satisfied when $\beta>1$ and $\lambda = \frac{\beta-1}{\beta}$.

(b) Neither WSP nor SSP is satisfied when $0<\beta<1$ and $\lambda = \frac{\beta-1}{\beta}$.

(c) SSP is satisfied when $\lambda = 0$ for any $\beta$.

(d) SSP is satisfied when $\lambda = \infty$ for any $\beta$.
\end{theorem}

\begin{proof}
The detailed proof can be found in Appendix \ref{proof of theorem 4}.
\end{proof}

\subsection{Fairness-efficiency tradeoff} 
\label{subsec:tradeoff}
In this subsection, we theoretically analyze the relationship between efficiency and fairness by examining the monotonicity of these two metrics.

For simplicity, we take $\alpha$-fairness function to discuss the tradeoff. As mentioned, when $\lambda = |\frac{1-\beta}{\beta}|$, the original optimization problem can be degenerated into the famous $\alpha$-fairness function in \eqref{alpha-fairness}, where $\alpha$ can be considered as the fairness measurement \cite{tang2006counter} \cite{renyi1961measures} \cite{menezes1970theory}. The larger $\beta$ ($\beta=\alpha$) is, the fairer the utility is. When $\beta=0$, it means our utility function in \eqref{alpha-fairness} only focuses on efficiency. When $\beta$ approaches $\infty$, our utility function becomes the fairest.

\begin{theorem} [Efficiency non-increasing property]
\label{theorem Efficiency non-increasing property}
Efficiency is non-increasing as $\beta$ increases if and only if when
\begin{equation}
\label{Efficiency non-increasing property}
    \sum^{M-K}_{i=1} \tau_i \det \overline{A_i} \ge 0.
\end{equation}
\end{theorem}

\begin{proof}
Let $\emph{R}$ be the $M \times K$ matrix that has full row rank. $M-K$ is the dimension of the null space. Let $z_i$ be the basis of the null space of $\emph{R}$, where $i = 1, 2,...,M-K$. Collect all $z_i$ to form the matrix $Z = [z_1, z_2,...,z_{M-K}]$. Then, we let $D$ be the Hessian matrix of our utility function in \eqref{alpha-fairness} and let $\mathbf{b} = \frac{\partial U}{\partial x \partial \beta}$. According to \cite{tang2006counter}, we can get that

\begin{equation}
\frac{\partial E}{\partial \beta} = \mathbf{1}^T Z(Z^T DZ)^{-1}Z^Tb.
\end{equation}

Let $A = Z^T D Z$, we can further infer that the conditions for that efficiency are non-increasing as $\alpha$ increases.
\end{proof}

We find efficiency is non-increasing with fairness as shown in Theorem \ref{theorem Efficiency non-increasing property}. In other words, there exists a tradeoff between efficiency and fairness when the conditions in Theorem \ref{theorem Efficiency non-increasing property} are satisfied. Furthermore, we discuss two scenarios that the tradeoff between efficiency and fairness doesn't exist. More specifically, we discuss the scenarios in which the allocation result can be the most efficient and fairest at the same time. As $\alpha$ is the measure of fairness, the fairest allocation scheme is max-min fairness when $\alpha$ approaches $\infty$. Thus, the equal-weighted dominant share is the fairest allocation result. To get the most efficient allocation result, it's essential to ensure as many resource constraints are tight as possible. We divide it into two scenarios: $K<M$ and $K=M$.

\begin{theorem} [Fairest and most efficient allocation scheme (\uppercase\expandafter{\romannumeral1})]
\label{counter-example11}
The equal-weighted dominant share scheme makes $K=M$ constraints tight if and only when 

\begin{equation}
    \sum_{i=1}^M\frac{\gamma_{i}^{<k>}}{\mu_i} = C,
\end{equation}
for any data analyst $i$ who requests $K$ resources, and $C$ is a constant.
\end{theorem}

\begin{proof}
The detailed proof can be found in Appendix \ref{proof of theorem 6}.
\end{proof}

\begin{theorem} [Fairest and most efficient allocation scheme (\uppercase\expandafter{\romannumeral2})]
\label{counter-example22}
The equal-weighted dominant share scheme makes $K<M$ constraints tight if and only if when at least one data analyst $i$'s dominant share $\mu_i x_i = 0$, which means at least one data analyst is allocated no dominant share. 
\end{theorem}

\begin{proof}
The detailed proof can be found in Appendix \ref{proof of theorem 7}.
\end{proof}

In the first scenario, DPBalance outputs equal share for each data analyst when each data analyst has the same weighted dominant share. In the second scenario, DPBalance outputs zero share for data analysts with zero weighted dominant share. Therefore, DPBalance covers these two special scenarios when conditions in Theorem \ref{counter-example11} or Theorem \ref{counter-example22} are satisfied.
\newcommand{\figwidthone}{0.31\linewidth}
\begin{figure*}[t]
\centering
\subfigure[$\beta=0.5$]{
\includegraphics[width=\figwidthone]{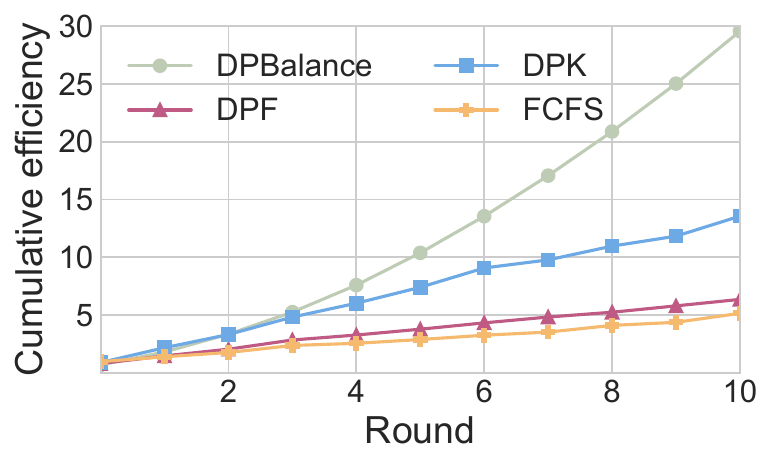}
}
\subfigure[$\beta=2.2$]{
\includegraphics[width=\figwidthone]{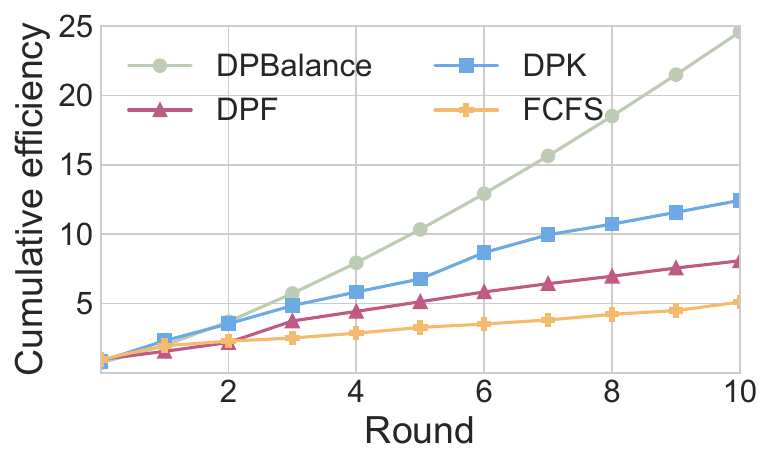}
}
\subfigure[$\beta=5.0$]{
\includegraphics[width=\figwidthone]{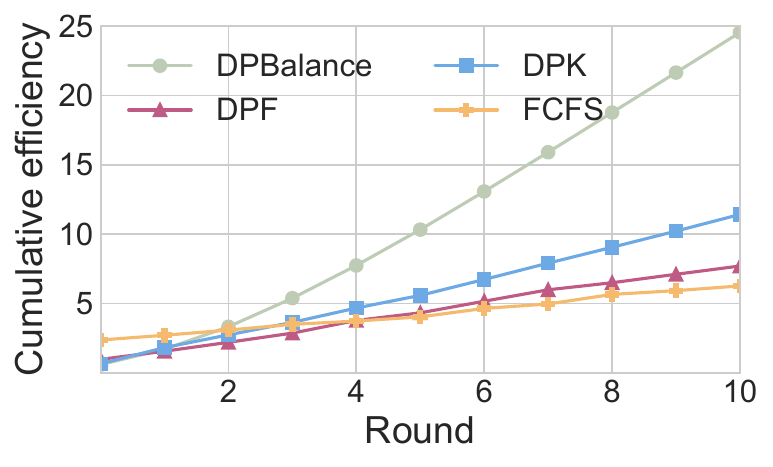}
}
\caption{Comparison of cumulative efficiency under different fairness preference settings}
\label{Cumulative efficiency}
\end{figure*}

\begin{figure*}[t]
\centering
\subfigure[$\beta=0.5$]{
\includegraphics[width=\figwidthone]{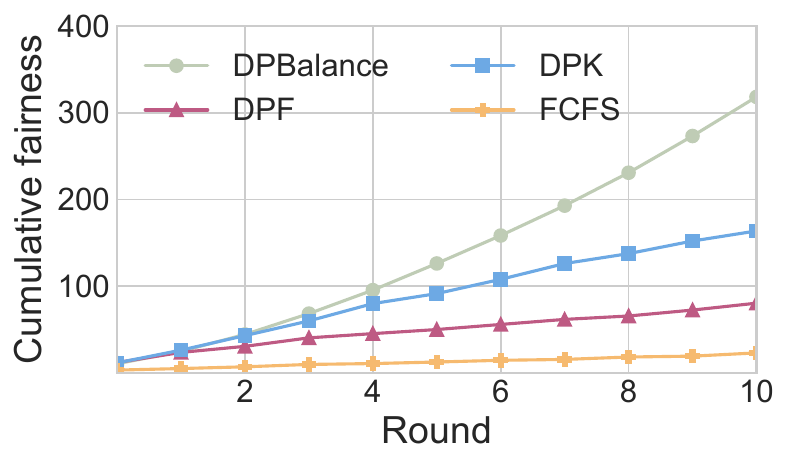}
}
\subfigure[$\beta=2.2$]{
\includegraphics[width=\figwidthone]{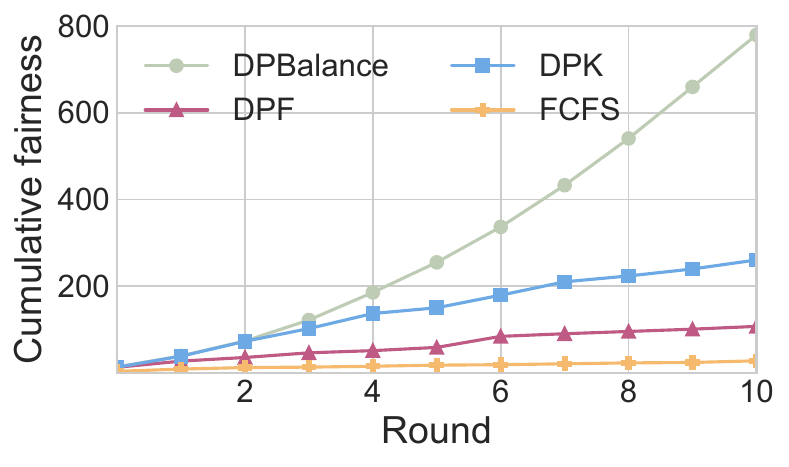}
}
\subfigure[$\beta=5.0$]{
\includegraphics[width=\figwidthone]{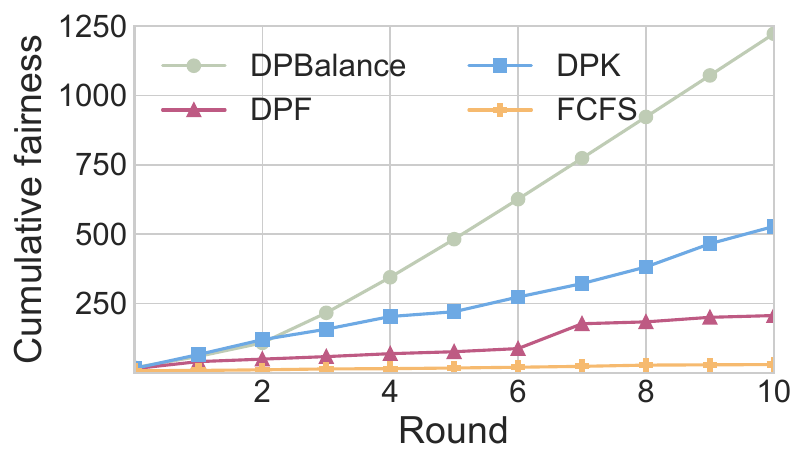}
}
\caption{Comparison of cumulative fairness under different fairness preference settings }
\label{Cumulative fairness}
\end{figure*}

\begin{figure}[ht]
\centering
\subfigure[Round level (random round)]{
\label{Round tradeoff (random round)}
\includegraphics[width=4.12cm]{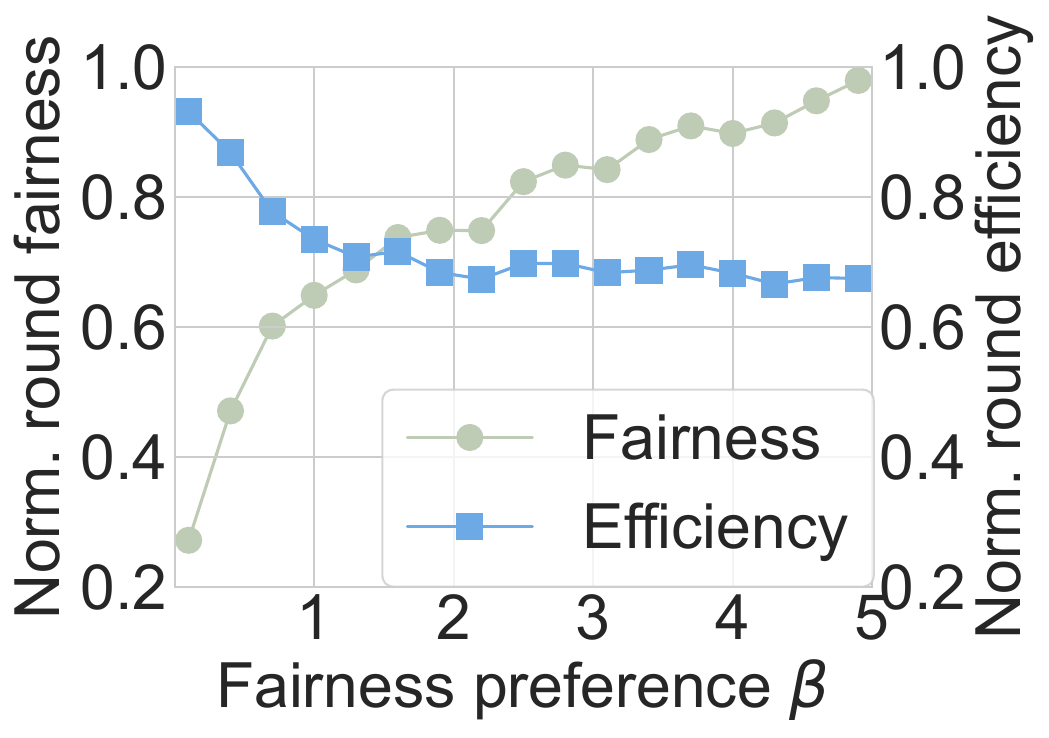}
}
\subfigure[Cumulative tradeoff (last round)]{
\label{Cumulative tradeoff (last round)}
\includegraphics[width=4.12cm]{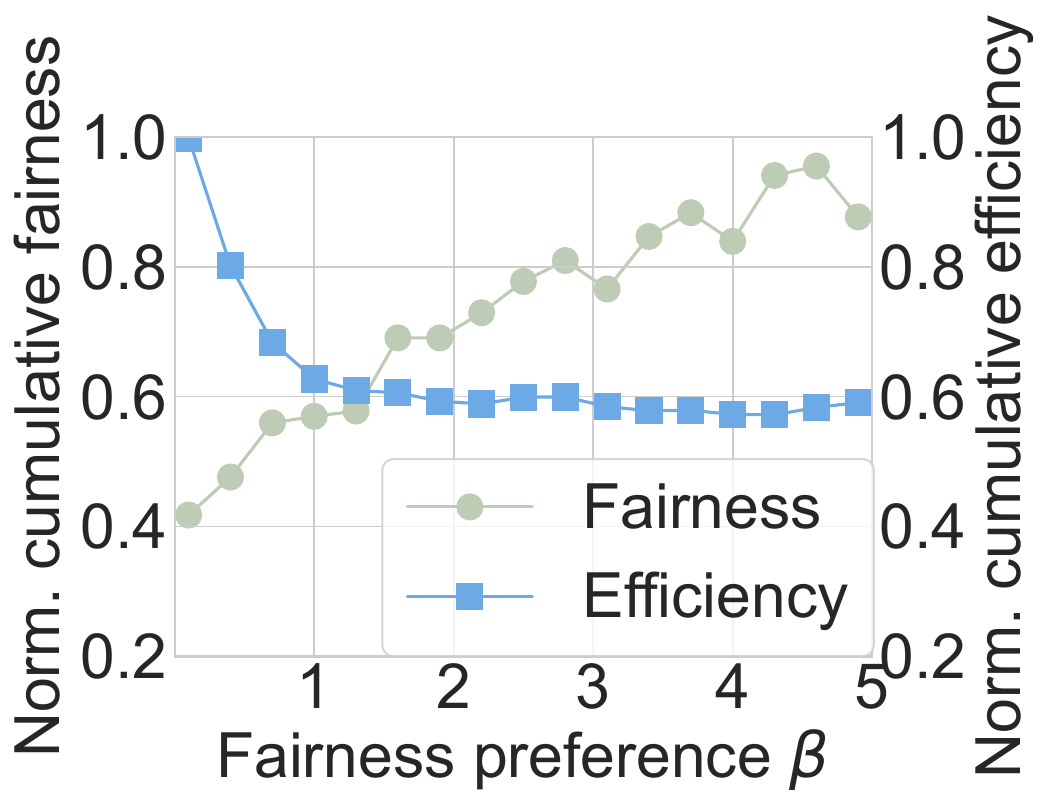}
}
\caption{Validation of efficiency-fairness tradeoff in two rounds}
\label{Round/ Cumulative Efficiency Fairness tradeoff}
\end{figure}

\section{Evaluation}
\label{Evaluation}
In this section, we evaluate DPBalance regarding its efficiency and fairness. We aim to answer the following questions:

\textbf{Q1:} How does DPBalance compare to other baseline scheduling methods?

\textbf{Q2:} How does fairness preference $\beta$ affect dominant efficiency and dominant fairness?

\textbf{Q3:} Can we empirically validate the existence of the tradeoff between dominant efficiency and dominant fairness when the condition of Theorem \ref{theorem Efficiency non-increasing property} is satisfied?

\textbf{Metrics:} We use the following four metrics to measure the performance of DPBalance and baselines.

\begin{itemize}
    
    \item \emph{Cumulative efficiency} is the sum of the dominant efficiency defined in \eqref{Dominant efficiency} of all rounds.
    
    \item \emph{Cumulative fairness} is the sum of the dominant fairness defined in \eqref{Dominant fairness} of all rounds.
    
    \item \emph{Round efficiency} is the dominant efficiency defined in \eqref{Dominant efficiency} of current training round.
   
    \item \emph{Round fairness} is the dominant fairness defined in \eqref{Dominant fairness} of current training round.
\end{itemize}

Round efficiency and round fairness are designed to testify if the tradeoff for a single round exists.

\textbf{Baselines:} We compare DPBalance with three baseline scheduling algorithms, DPK\cite{tholoniat2022packing}, DPF \cite{luo2021privacy} and First-Come-First-Serve (FCFS). DPK allocates resources to the pipelines with lowest weight-to-demand ratio. DPF allocates resources to pipelines with smallest dominant share. 

First-Come-First-Serve (FCFS) allocates resource to pipelines who arrive first.

\textbf{Simulation Setup:}
For the resource side, we create $100$ edge devices and each device's global privacy budget $\epsilon_g$ follows the Uniform distribution, namely $\epsilon_g \sim U(1.0, 1.5)$. For each device, two new blocks are created every 10 seconds. For the data analysts side, we assume 6 data analysts, each of whom with 25 pipelines, come every 10 seconds following the Poisson process. Depending on the amount of demanding privacy resources, we generate two types of pipelines, 75\% mice and 25\% elephant pipelines following the setting in DPF \cite{luo2021privacy}. For mice pipeline, it demands data blocks with privacy demand $\epsilon \thicksim U(0.005, 0.015)$. For elephant pipeline, it demands data blocks with privacy demand $\epsilon \thicksim U(0.095, 0.105)$. For each pipeline, it demands the latest 10 data blocks with a probability of 0.25 or the latest 1 data block with a probability of 0.75. After being generated, all pipelines are shuffled and put into different data analysts, each of whom consists of $25$ pipelines. For each data analyst, it either demands $0.2$ of all devices or all the devices with probability $0.5$ and $0.5$ respectively.

\newcommand{\figsize}{0.30\linewidth}

\subsection{Comparison with baselines (\textbf{Q1})}
We compare DPBalance with other baseline scheduling methods. We show the cumulative round efficiency and cumulative round fairness in Fig. \ref{Cumulative efficiency} and Fig. \ref{Cumulative fairness} respectively. 

We choose three different fairness preference $\beta$ to represent efficiency-preferred ($\beta = 0.5$), efficiency-fairness-unbiased ($\beta = 2.2$), fairness-preferred ($\beta = 5.0$) situation. We take the last round (round 10) as an example. For efficiency, DPBalance outperforms than other three baselines by improving $1.44 \times \sim 3.49 \times$, $1.41\times \sim 3.04 \times$, $1.67\times \sim 2.28 \times$ on average under the efficiency-preferred, efficiency-fairness-unbiased and fairness-preferred situation respectively. For fairness, DPBalance outperforms the other three baselines by improving $1.37 \times \sim 9.66 \times$, $1.77\times \sim 15.58 \times$, $1.82\times \sim 24.32 \times$ on average under the efficiency-preferred, efficiency-fairness-unbiased and fairness-preferred situation respectively. Therefore, we can get DPBalance outperforms the other three scheduling baselines on cumulative round efficiency and cumulative round fairness under different fairness preferences.

\subsection{Impact of fairness preference $\beta$ (\textbf{Q2})}

We show the round efficiency, round fairness, cumulative efficiency, and cumulative fairness as fairness preference $\beta$ increases in Fig. \ref{Round/ Cumulative Efficiency Fairness tradeoff}. As fairness preference $\beta$ grows from $0$ to $5$, which means allocation becomes more fair-preferred, round efficiency and cumulative efficiency decrease $48\%$ and $38\%$ respectively. As fairness preference $\beta$ grows from $0$ to $5$, round fairness and cumulative fairness increase $257\%$ and $119\%$ respectively. Therefore, $\beta$ can serve as a controlling knob for the system manager to adjust fairness and efficiency.

\subsection{Validation of tradeoff (\textbf{Q3})}

The dominant efficiency and dominant fairness are normalized for a better comparison. We present the results in one randomly selected round in Fig. \ref{Round tradeoff (random round)}. Round efficiency decreases and round fairness increases as fairness preference $\beta$ increases from $0$ to $5$. In addition, we show the cumulative efficiency and cumulative fairness of the last round in Fig. \ref{Cumulative tradeoff (last round)}. Similarly, cumulative fairness increases, and cumulative efficiency decreases as fairness preference $\beta$ increases. Therefore, the tradeoff between dominant efficiency and dominant fairness does exist when conditions in Theorem \ref{theorem Efficiency non-increasing property} are satisfied, which verifies the proof of the tradeoff in Section \ref{subsec:tradeoff}.

\section{Conclusion}
\label{Conclusion}
FLaaS is an emerging cloud-based service that facilitates FL-based model training over crowdsourced data owners. Due to the intrinsic properties of DP, privacy budget set by data owners have to be carefully allocated to different data analysts to improve system efficiency and data analyst fairness. Existing privacy budget scheduling mechanisms consider either efficiency or fairness only. This paper completes the missing piece by proposing a new mechanism, DPBalance, that jointly optimizes efficiency and fairness. We prove theoretically that DPBalance guarantees four key economic properties and there exists an fairness-efficiency tradeoff in practical situations. Extensive experiments further demonstrate the superiority of DPBalance in efficiency and fairness performance.

\begin{appendices}
\section{proof of theorem 1}
\label{proof of theorem 1}
The first optimization problem in Eq.(\ref{global opt}) along with first constraint is equal to \begin{equation}
\begin{aligned}
    &\max \Phi_{\lambda}(x) = l\Big(f_{\beta}(x)\Big) + \lambda l\Big(\sum_{i=1}^m \mu_i x_i T(t_i) l_i \Big), \\
    &l(y) = sgn(y) log(|y|), \\
    &f_{\beta} = sgn(1-\beta)\bigg(\sum_{i=1}^m \Big(\frac{\mu_i x_i T(t_i) l_i}{\sum_{i=1}^m \mu_i x_i T(t_i) l_i} \Big)^{1-\beta} \bigg)^{\frac{1}{\beta}}.\\
\end{aligned}
\end{equation}
We first set $\beta > 1$ and other cases can be proved by symmetry. Let $x$ and $x'$ be 2 vectors representing 2 allocation schemes, and set $x'$ Pareto-dominates $x$. Then, we can get

\begin{equation}
    \begin{aligned}
        & \Phi_{\lambda}(x')-\Phi_{\lambda}(x) \\
        & = l\Big(f_{\beta}(x')\Big) - l\Big(f_{\beta}(x)\Big) \\
        &+ \lambda \bigg(l\Big(\sum_{i=1}^m \mu_i x'_i T(t_i) l_i\Big) - l\Big(\sum_{i=1}^m \mu_i x_i T(t_i) l_i\Big) \bigg)\\
        & = - \Bigg(log\Big(f_{\beta}(x')\Big) - log\Big(f_{\beta}(x)\Big) \Bigg)\\
        &+ \lambda \bigg(log\Big(\sum_{i=1}^m \mu_i x'_i T(t_i) l_i\Big) - log\Big(\sum_{i=1}^m \mu_i x_i T(t_i) l_i\Big) \bigg)\\
        & = - \Bigg(log\Big(f_{\beta}(x')\Big) - log\Big(f_{\beta}(x)\Big) \Bigg)\\
        &+ \lambda \bigg(log\Big((1+\delta)\sum_{i=1}^m \mu_i x_i T(t_i) l_i\Big) - log\Big(\sum_{i=1}^m \mu_i x_i T(t_i) l_i\Big) \bigg)\\
        & = - \Bigg(log\Big(f_{\beta}(x')\Big) - log\Big(f_{\beta}(x)\Big) \Bigg) + \lambda(1+\delta). \\
    \end{aligned}
\end{equation}
(a) If $x'$ is more fair than $x$, we can get $f_{\beta}(x') > f_{\beta}(x)$ because $f_{\beta}$ is a measure function of fairness. Thus, we can get
\begin{equation}
\begin{aligned}
    &- \bigg(log\Big(f_{\beta}(x')\Big) - log\Big(f_{\beta}(x)\Big) \bigg) > 0 \\
    & \Phi_{\lambda}(x')-\Phi_{\lambda}(x) \\ 
    & = - \bigg(log\Big(f_{\beta}(x')\Big) - log\Big(f_{\beta}(x)\Big) \bigg) + \lambda(1+\delta) > 0 \\
\end{aligned}
\end{equation}

(b) If $x'$ is less fair than $x$,
\begin{equation}
    \begin{aligned}
        & \Phi_{\lambda}(x')-\Phi_{\lambda}(x) \\ 
        & = -log \Big( \sum_{i=1}^m \big( \frac{\mu_i x'_i T(t_i) l_i}{\sum_{i=1}^m \mu_i x'_i T(t_i) l_i} \big)^{1-\beta} \Big)^{\frac{1}{\beta}}  \\
        & +log \Big( \sum_{i=1}^m \big( \frac{\mu_i x_i T(t_i) l_i}{\sum_{i=1}^m \mu_i x_i T(t_i) l_i} \big)^{1-\beta} \Big)^{\frac{1}{\beta}} + \lambda log(1+\delta). \\
    \end{aligned}
\end{equation}

Let $a_i = \mu_i x_i T(t_i) l_i$, $a'_i = \mu_i x'_i T(t_i) l_i$. Thus, we get $a'_i = (1+\delta)a_i$. Thus, we get
\begin{equation}
    \begin{aligned}
        & \Phi_{\lambda}(x')-\Phi_{\lambda}(x) \\ 
        & = -\frac{1}{\beta} log\Bigg(\frac{\sum_{i=1}^m \bigg(\frac{a'_i}{\sum_{i=1}^m a'_i}\bigg)^{1-\beta}}{\sum_{i=1}^m \bigg(\frac{a_i}{\sum_{i=1}^m a_i}\bigg)^{1-\beta}}\Bigg) + \lambda log(1+\delta). \\
        & = -\frac{1}{\beta} log\Bigg(\frac{\sum_{i=1}^m \bigg(\frac{a'_i}{\sum_{i=1}^m a_i(1+\delta)}\bigg)^{1-\beta}}{\sum_{i=1}^m \bigg(\frac{a_i}{\sum_{i=1}^m a_i}\bigg)^{1-\beta}}\Bigg) + \lambda log(1+\delta). \\
        & = -\frac{1}{\beta} log\Bigg( \frac{\sum_{i=1}^m (a'_i)^{1-\beta}}{\sum_{i=1}^m (a_i)^{1-\beta}}  \cdot \bigg( \frac{1}{1+\delta} \bigg)^{1-\beta} \Bigg) + \lambda log(1+\delta). \\
        & = -\frac{1}{\beta} log\Bigg( \frac{\sum_{i=1}^m (a'_i)^{1-\beta}}{\sum_{i=1}^m (a_i)^{1-\beta}} \Bigg) -\frac{\beta-1}{\beta} log(1+\delta) + \lambda log(1+\delta). \\
        & = -\frac{1}{\beta} log\Bigg( \frac{\sum_{i=1}^m (a'_i)^{1-\beta}}{\sum_{i=1}^m (a_i)^{1-\beta}} \Bigg) + (\lambda - \frac{\beta-1}{\beta})log(1+\delta). \\
    \end{aligned}
\end{equation}
We have $a'_i \ge a_i$ and $\beta > 1$. Thus,
\begin{equation}
\begin{aligned}
    & (a'_i)^{1-\beta} \le (a_i)^{1-\beta}\\
    & -\frac{1}{\beta} log\Bigg( \frac{\sum_{i=1}^m (a'_i)^{1-\beta}}{\sum_{i=1}^m (a_i)^{1-\beta}} \Bigg) > 0\\
\end{aligned}
\end{equation}

When ${\lambda}  \ge {\frac{1-\beta}{\beta}}$, we have 
\begin{equation}
    (\lambda - \frac{\beta-1}{\beta})log(1+\delta)>0.
\end{equation}

Thus, we can get $\Phi_{\lambda}(x')-\Phi_{\lambda}(x) \ge 0$.

For the completeness of proof, we have to prove if $(\lambda - \frac{\beta-1}{\beta})log(1+\delta)>0$ would let $\Phi_{\lambda}(x')-\Phi_{\lambda}(x)<0$. Now we consider a new case: We extend the length of the vector $x$ to $m+1$. Let $x_i=1$, $x_{m+1} = m$, $x'_i=x_i$, $x'_{m+1}=x_{m+1}+2\delta(\sum_{i=1}^m x_i)$, where $i \in [1,m]$. Thus, we can get
\begin{equation}
\begin{aligned}
    & \Phi_{\lambda}(x') - \Phi_{\lambda}(x) \\
    & = -\frac{1}{\beta} log\Bigg( \frac{\sum_{i=1}^{m+1} (x'_i)^{1-\beta}}{\sum_{i=1}^{m+1} (x_i)^{1-\beta}} \Bigg) + (\lambda - \frac{\beta-1}{\beta})log(1+\delta). \\
    & = -\frac{1}{\beta} log\Bigg( \frac{n+(n+2n\delta)^{1-\beta}}{n+n^{1-\beta}} \Bigg) + (\lambda - \frac{\beta-1}{\beta})log(1+\delta). \\
    & \le -\frac{1}{\beta} log\Bigg( \frac{n}{n+n^{1-\beta}} \Bigg) + (\lambda - \frac{\beta-1}{\beta})log(1+\delta). \\
    & = -\frac{1}{\beta} log\Bigg( \frac{1}{1+n^{-\beta}} \Bigg) + (\lambda - \frac{\beta-1}{\beta})log(1+\delta). \\
\end{aligned}
\end{equation}

Let $\delta = \frac{1}{2} \bigg( \Big( 1+n^{-\beta} \Big)^{-\frac{1}{\beta} \frac{1}{(\lambda-\frac{\beta-1}{\beta})}} \bigg)>0.$
Under this condition, we can get
\begin{equation}
    \Phi_{\lambda}(x') - \Phi_{\lambda}(x) \le -\frac{1}{\beta} log\Bigg( \frac{1}{1+n^{-\beta}} \Bigg) + (\lambda - \frac{\beta-1}{\beta})log(1+\delta) <0. 
\end{equation}
To keep first optimization problem in Eq.(\ref{global opt}) positive, it's essential to add the absolute symbol. Thus, to keep Pareto Efficiency, we have
\begin{equation}
    \beta>0, \vert{\lambda} \vert \ge \vert{\frac{1-\beta}{\beta}} \vert.
\end{equation}

\section{proof of theorem 2}
\label{proof of theorem 2}
(a) When $\beta > 1$ and $\lambda = \frac{\beta-1}{\beta}$, the original optimization problem in Eq.(\ref{opt 1}) can be transformed into Eq.(\ref{alpha-fairness}).

Combined with first constraint in Eq.(\ref{global opt}), the Lagrangian function can be represented as

\begin{equation}
    L = \sum_{i=1}^m \frac{(\mu_i x_i T(t_i) l_i)^{1-\beta}}{1-\beta} - \sum_{k=1}^K \lambda_k (\sum_{i=1}^m \gamma_{ik}x_i T(t_i) l_i - 1).
\end{equation}

We can get 
\begin{equation}
\label{e27}
    \Big( \mu_i T(t_i) l_i \Big)^{1-\beta} x
    ^{-\beta} = \sum_{k=1}^K \lambda_k \gamma_{ik} T(t_i) l_i, \forall k,
\end{equation}

and

\begin{equation}
\label{e28}
    \sum_{k=1}^K \sum_{i=1}^m \lambda_k \gamma_{ik} x_i = \sum_{k=1}^K \lambda_k.
\end{equation}

From Eq.(\ref{e27}), we can get

\begin{equation}
\label{e29}
    \Big( \mu_i T(t_i) l_i \Big)^{-\beta} x^{-\beta} = \sum_{k=1}^K \lambda_k \frac{\gamma_{ik}}{\mu_i} \le \sum_{k=1}^K \lambda_k.
\end{equation}

Combine Eq.(\ref{e27}) and Eq.(\ref{e28}), we can get

\begin{equation}
    \sum_{k=1}^K \sum_{i=1}^m \lambda_k \gamma_{ik} x_i = \sum_{i=1}^m (\mu_i T(t_i)l_i)^{1-\beta} x^{1-\beta}.
\end{equation}

and

\begin{equation}
    \sum_{i=1}^m (\mu_i T(t_i)l_i)^{1-\beta} x^{1-\beta} = \sum_{k=1}^K \lambda_k.
\end{equation}

Combine with Eq.(\ref{e29}), we can get
\begin{equation}
    \Big( \mu_i T(t_i) l_i \Big)^{-\beta} x_i^{-\beta} \le \sum_{i=1}^m \Big( \mu_i T(t_i) l_i \Big)^{1-\beta} x_i^{1-\beta}.
\end{equation}

Let $a_i = T(t_i) l_i$, then we can get
\begin{equation}
    (\mu_i a_i)^{-\beta} x_i^{-\beta} \le \sum_{i=1}^m (\mu_i a_i)^{1-\beta} x^{1-\beta}.
\end{equation}
Thus, we can get

\begin{equation}
\begin{aligned}
    &\min_i a_i \mu_i x_i \ge \frac{1}{m}\\
    &\min_i = \mu_i x_i T(t_i) l_i \ge \frac{1}{m}.
\end{aligned}
\end{equation}
Thus, under this condition, Sharing Incentive property is satisfied.

(b) When $0<\beta<1$ and $\lambda = \frac{\beta-1}{\beta}$, we use $\gamma_{i}$ to represent $\gamma_{ik}$ and base on Eq.(\ref{e27}) we can get 
\begin{equation}
\label{e35}
    \Big( \mu_1 T(t_1) l_1 \Big)^{1-\beta} x_1^{-\beta} = \sum_{k=1}^K \lambda_k \gamma_1 T(t_1) l_1.
\end{equation}

From Eq.(\ref{e27}) and Eq.(\ref{e35}), we can get

\begin{equation}
    x_i = x_1 \Big( \frac{\gamma_i a_i}{\gamma_1 a_1} \Big)^{-\frac{1}{\beta}} \Big( \frac{\mu_1 a_1}{\mu_i a_i} \Big)^{\frac{\beta-1}{\beta}},
\end{equation}

where $a_i = T(t_i) l_i$.

Combine with Eq.(\ref{resource availability}), we get
\begin{equation}
\label{e37}
    x_i = \frac{(\gamma_i a_i)^{-\frac{1}{\beta}}}{(\mu_i a_i)^{\frac{\beta-1}{\beta}} \sum_{i=1}^m \Big(\frac{\gamma_i}{\mu_i}\Big)^{\frac{\beta-1}{\beta}}}
\end{equation}

Thus, we can get

\begin{equation}
    \mu_i a_i x_i = \frac{(\gamma_i a_i)^{-\frac{1}{\beta}}}{(\mu_i a_i)^{-\frac{1}{\beta}} \sum_{i=1}^m \Big(\frac{\gamma_i}{\mu_i}\Big)^{\frac{\beta-1}{\beta}}} < \frac{1}{m}.
    \label{e43}
\end{equation}
Therefore, there exists some data analysts whose dominant share of allocation is less than equal share, which doesn't satisfy Sharing Incentive property.

(c) When $\lambda =0$ for any $\beta$, which means there only left fairness part of utility. In this case, each data analyst would only get equal weighted dominant share, which means $\sum_{i=1}^m a_i \mu_i x_i >1$ and there is no data analyst whose weighted dominant share is less than $\frac{1}{m}$. Thus, Sharing Incentive property should be satisfied.

(d) When $|\lambda| = \infty$ for any $\beta$, which means efficiency would be heavily weighted. 

Suppose there is a case: 2 data analysts require 2 data blocks. To simplify the case, we assume $T(t_i)=1$ and $l_i = 1$. And data analyst 1's normalized demand $D_1=(1, \gamma), \gamma < 1$, while data analyst 2's normalized demand $D_2 = (0, 1)$. Thus, we can get $\mu_1 = \mu_2 = 1$. As $\lambda = \infty$, we can simply ignore the fairness part. Thus, the original optimization problem are equal to 
\begin{equation}
    \max x_1 + x_2.
\end{equation}
And we can get the optimal solution $x_1 = 1, x_2 = 1-\gamma$ and the dominant share of them is $\mu_1 x1 = 1, \mu_2 x_2 = 1-\gamma$. When $\gamma > \frac{1}{2}$, we can get $\mu_2 x_2 < \frac{1}{2}=\frac{1}{n}$, which make Sharing Incentive dissatisfied.

\section{proof of theorem 3}
\label{proof of theorem 3}
(a) When $\beta > 1$ and $\lambda = \frac{\beta-1}{\beta}$, we assume data analyst i envies data analyst j, namely $\mu_j x_j T(t_j) l_j > \mu_i x_i T(t_i) l_i$. Thus, we can get at least there exist one resource $k$,
\begin{equation}
\label{e40}
    \gamma_{jk} x_j T(t_j) l_j \ge \gamma_{ik} x_i T(t_i) l_i. 
\end{equation}

Combine with Eq.(\ref{e27}), we can get
\begin{equation}
\begin{aligned}
    & \Big( \mu_i T(t_i) l_i \Big)^{1-\beta} x_i
    ^{1-\beta} = \sum_{k=1}^K \lambda_k \gamma_{ik} x_i T(t_i) l_i \\
    & < \sum_{k=1}^K \lambda_k \gamma_{jk} x_j T(t_j) l_j = \Big( \mu_j T(t_j) l_j \Big)^{1-\beta} x_j^{1-\beta}.
\end{aligned}
\end{equation}
As $\beta >1$, we can get
\begin{equation}
    \gamma_{jk} x_j T(t_j) l_j < \gamma_{ik} x_i T(t_i) l_i,
\end{equation}
which is opposite to Eq.(\ref{e40}). Thus, under this condition, Envy-Freeness is satisfied.

(b) When $0<\beta<1$ and $\lambda = \frac{\beta-1}{\beta}$, we can get $\gamma_{jk} x_j T(t_j) l_j > \gamma_{ik} x_i T(t_i) l_i$ under case in (a). Thus, under this condition, Envy-Freeness is not satisfied.

(c) When $\lambda =0$ for any $\beta$, which means there only left fairness part of utility. Under this condition, each data analyst gets equal weighted dominant share,

\begin{equation}
    \mu_1 x_1 T(t_1) l_1 = \mu_2 x_2 T(t_2) l_2 =...= \mu_m x_m T(t_m) l_m.
\end{equation}
Thus, under this condition, Envy-Freeness is satisfied.

(d) When $|\lambda| = \infty$ for any $\beta$, which means efficiency would be heavily weighted. In the same case of (d) in Appendix B. As data analyst 1's dominant share is 1 while data analyst 2's dominant share is $1-\gamma$. Thus, under this condition, Envy-Freeness is not satisfied.

\begin{figure*}[ht]
\centering
\subfigure[An counter-example for $K=M$: 2 data analysts request for 2 data blocks.]{
\label{counter-example for K=M}
\includegraphics[width=5.0cm]{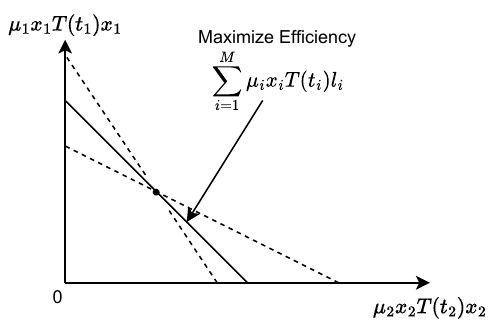}
}
\subfigure[An counter-example for $K<M$: 2 data analysts request for 2 data blocks. Left one represents that one resource constraint is tight and the optimal allocation is unique; Right one represents that one resource constraint is tight and the optimal allocation is not unique.]{
\label{counter-example for K<M}
\includegraphics[width=9cm]{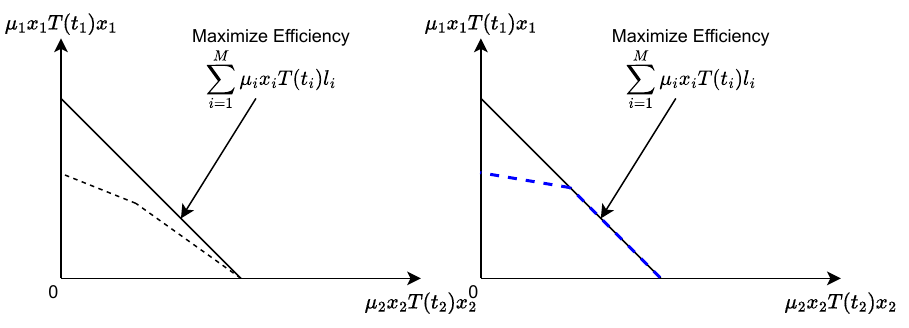}
}
\caption{Counter-examples}
\label{Counter-example}
\end{figure*}

\section{proof of theorem 4}
\label{proof of theorem 4}
Let's assume that the data analyst has three basic ways of lying.

(1) ($\mu_i' > \mu_i$) Data analysts only lie about more dominant share in the training demands. 

(2) ($\gamma_i' > \gamma_i$) Data analysts only lie about more non-dominant share in the training demands. 

(3) ($\mu_i' > \mu_i, \gamma_i' > \gamma_i, \frac{\gamma_i'}{\mu_i'} = \frac{\gamma_i}{\mu_i}$) Data analysts lie about more dominant share and non-dominant share in proportion in the training demands. 

For more lying ways, they can be decomposed into 3 basic ways mentioned above. For example, if data analysts lie about more dominant share and non-dominant share out of proportion, it can be decomposed into (1) + (3) or (2) + (3).

For simplicity, we use $\mu$ and $\gamma$ to represent dominant share $\mu x$ and $\gamma x$ because $x=1$ can meet some pipeline's training demand. Now we consider 4 different scenarios:

(a) When $\beta > 1$ and $\lambda = \frac{\beta-1}{\beta}$, from Eq.(\ref{e43}), we can get

\begin{equation} 
\begin{aligned}
\label{e44}
    & \mu_j a_j x_j = \frac{(\frac{\gamma_j}{\mu_j})^{-\frac{1}{\beta}}}{ \sum_{i=1}^m \Big(\frac{\gamma_i}{\mu_i}\Big)^{\frac{\beta-1}{\beta}}} \\
    & = \frac{\gamma_j ^{-\frac{1}{\beta}}}{\frac{\gamma_j}{\mu_j} (\gamma_j)^{-\frac{1}{\beta}}  \sum_{i=1, i \neq j}^m \frac{\gamma_i}{\mu_i} (\gamma_i)^{-\frac{1}{\beta}} \Big( \frac{\mu_j}{\mu_i} \Big)^{-\frac{1}{\beta}}}.\\
\end{aligned}
\end{equation}

If data analyst $j$ lies about more dominant share $\mu_j$, it would get more weighted dominant share $\mu_j a_j x_j$.  

Combine with Eq.(\ref{e37}), we can get 

\begin{equation}
\label{e45}
    \begin{aligned}
        & \gamma_j a_j x_j = \frac{(\frac{\gamma_j}{\mu_j})^{\frac{\beta-1}{\beta}}}{ \sum_{i=1}^m \Big(\frac{\gamma_i}{\mu_i}\Big)^{\frac{\beta-1}{\beta}}} \\
        & = \frac{(\gamma_j)^{\frac{\beta-1}{\beta}}}{(\gamma_j)^{\frac{\beta-1}{\beta}} + \mu_j^{\frac{\beta-1}{\beta}} \sum_{i=1, i \neq j}^m  \Big(\frac{\gamma_i}{\mu_i}\Big)^{\frac{\beta-1}{\beta}}}. \\
    \end{aligned}
\end{equation}

Because $\beta > 1$, we can get if $\mu_j$ increases, data analyst $j$ would get less weighted non-dominant share $\gamma_j a_j x_j$. 

If data analyst $j$ lies about more non-dominant share $\gamma_j$, it wouldn't increase its utility. Even it gets more non-dominant share, it still wouldn't get more dominant share and increase its utility because of Definition 4.

If data analyst $j$ lies about more dominant share $\mu_j$ and non-dominant share $\gamma_j$ in proportion, it wouldn't increase its utility. As shown in Eq.(\ref{e44}) and Eq.(\ref{e45}), if $\frac{\gamma_j}{\mu_j}$ keeps fixed, $\mu_j a_j x_j$ and $\gamma_j a_j x_j$ also keep fixed.

Thus, under this condition, only weak Strategy Proofness property is satisfied. Only when data analyst lies more about its dominant share can it gets more dominant share, but it also puts data analyst at risk of losing non-dominant share.

(b) When $0<\beta<1$ and $\lambda = \frac{\beta-1}{\beta}$, if data analyst $j$ lies about more dominant share $\mu_j$, it would get more $\mu_j a_j x_j$ and $\gamma_j a_j x_j$ from Eq.(\ref{e44}) and Eq.(\ref{e45}).
Thus, both strong Strategy Proofness property and weak Strategy Proofness property are not satisfied. 

(c) When $\lambda =0$ for any $\beta$, which means there only left fairness part of utility. In this case, each data analyst gets equal weighted dominant share,
$$\mu_1 a_1 x_1 = \mu_2 a_2 x_2 =...= \mu_m a_m x_m.$$
For any data analyst $i$, we can get $\mu_i a_i x_i = \mu_1 a_1 x_1$. 

If data analyst $j$ only lies about more dominant share $\mu_j$, it would decrease $x_j$ and $\mu_j a_j x_j$ remains the same. Thus, it would gets more dominant share.

If data analyst $i$ only lies about more non-dominant $\gamma_j$, it wouldn't increase it utility even it gets more non-dominant share because of Definition 4.

If data analyst lies more about dominant share $\mu_j$ and $\gamma_j$ in proportion, $\mu_i a_i x_i$ remains the same because of $\mu_i a_i x_i = \mu_1 a_1 x_1$.

Thus, Weak Strategy Proofness property is satisfied.

(d) When $|\lambda| = \infty$ for any $\beta$, which means efficiency would be heavily weighted. We ignore the fairness part, and the optimization problem equals to 
\begin{equation}
\begin{aligned}
    & \max \sum_{i=1}^m \mu_i a_i x_i, \\
    & s.t. \sum_{i=1}^m \gamma_{ik}x_i T(t_i) l_i \le 1, k\in[1,K]. \\
\end{aligned}
\end{equation}

And the optimal solution is 
$$\mu_1 a_1 x_1 = \mu_2 a_2 x_2 =...= \mu_m a_m x_m.$$

From (c), we know Weak Strategy Proofness is satisfied under this condition.

\section{proof of theorem 6}
\label{proof of theorem 6}
From the resource constraints in Eq.(\ref{resource availability}), we can get $\sum_{i=1}^M \gamma_{ik} x_i T(t_i) l_i = 1$ for any resource $k$ if $K=M$ constraints are tight. Also, because of the equal weighted dominant share under this scenario, we let $\mu_i x_i T(t_i) l_i = H$, where $H$ is a constant. Thus, we can get
\begin{equation}
\label{e50}
    \sum_{i=1}\frac{\gamma_{ik}}{\mu_i} = \frac{1}{H}.
\end{equation}
We give the 2 data analysts counter example in Fig.\ref{counter-example for K=M}.

\section{proof of theorem 7}
\label{proof of theorem 7}
From the resource constraints in Eq.(\ref{resource availability}) and $\mathbf{x} \ge 1$, we can get the solution space is a convex hull bounded by these constraints. Also, because of the equal dominant share, The optimal solution either intersects a convex hull on a surface or is relative to a point on the axis in order to maximize the efficiency. Thus, we can get at least $M-K$ data analysts are allocated no weighted dominant share. We give the 2 data analysts counter example in Fig.\ref{counter-example for K<M}.

\end{appendices}

\clearpage
\input{conference.bbl}

\end{document}

%% file: conference.bbl

%% file: conference.bbl
\begin{thebibliography}{10}
\providecommand{\url}[1]{#1}
\csname url@samestyle\endcsname
\providecommand{\newblock}{\relax}
\providecommand{\bibinfo}[2]{#2}
\providecommand{\BIBentrySTDinterwordspacing}{\spaceskip=0pt\relax}
\providecommand{\BIBentryALTinterwordstretchfactor}{4}
\providecommand{\BIBentryALTinterwordspacing}{\spaceskip=\fontdimen2\font plus
\BIBentryALTinterwordstretchfactor\fontdimen3\font minus
  \fontdimen4\font\relax}
\providecommand{\BIBforeignlanguage}[2]{{%
\expandafter\ifx\csname l@#1\endcsname\relax
\typeout{** WARNING: IEEEtran.bst: No hyphenation pattern has been}%
\typeout{** loaded for the language `#1'. Using the pattern for}%
\typeout{** the default language instead.}%
\else
\language=\csname l@#1\endcsname
\fi
#2}}
\providecommand{\BIBdecl}{\relax}
\BIBdecl

\bibitem{GDPR_2018EU}
\BIBentryALTinterwordspacing
{European Commission}, ``General data protection regulation,'' May 2018.
  [Online]. Available: \url{https://gdpr-info.eu/}
\BIBentrySTDinterwordspacing

\bibitem{chen2019federated}
M.~Chen, R.~Mathews, T.~Ouyang, and F.~Beaufays, ``Federated learning of
  out-of-vocabulary words,'' \emph{arXiv}, Apr. 2019.

\bibitem{ramaswamy2019federated}
S.~Ramaswamy, R.~Mathews, K.~Rao, and F.~Beaufays, ``Federated learning for
  emoji prediction in a mobile keyboard,'' \emph{arXiv}, Jun. 2019.

\bibitem{tim2018applied}
T.~Yang, G.~Andrew, H.~Eichner, H.~Sun, W.~Li, N.~Kong, D.~Ramage, and
  F.~Beaufays, ``Applied federated learning: Improving google keyboard query
  suggestions,'' \emph{arXiv}, Jan. 2018.

\bibitem{leroy2019federated}
D.~Leroy, A.~Coucke, T.~Lavril, T.~Gisselbrecht, and J.~Dureau, ``Federated
  learning for keyword spotting,'' in \emph{Proc. {IEEE} Int. Conf. Acoust.,
  Speech Signal Process. (ICASSP)}, May 2019, pp. 6341--6345.

\bibitem{de2019federated}
\BIBentryALTinterwordspacing
W.~De~Brouwer, ``The federated future is ready for shipping,'' Mar. 2019.
  [Online]. Available:
  \url{https://medium.com/@_doc_ai/the-federated-future-is-ready-for-shipping-d17ff40f43e3}
\BIBentrySTDinterwordspacing

\bibitem{he2020fedml}
C.~He, S.~Li, J.~So, M.~Zhang, H.~Wang, X.~Wang, P.~Vepakomma, A.~Singh,
  H.~Qiu, L.~Shen, P.~Zhao, Y.~Kang, Y.~Liu, R.~Raskar, Q.~Yang, M.~Annavaram,
  and S.~Avestimehr, ``Fedml: {A} research library and benchmark for federated
  machine learning,'' \emph{arXiv}, Jun. 2020.

\bibitem{FATE2019}
\BIBentryALTinterwordspacing
{The FATE Authors}, ``Federated ai technology enabler,'' 2019. [Online].
  Available: \url{https://www.fedai.org/}
\BIBentrySTDinterwordspacing

\bibitem{FLaaS2020Nicolas}
N.~Kourtellis, K.~Katevas, and D.~Perino, ``Flaas: Federated learning as a
  service,'' in \emph{Proc. Workshop Distrib. Mach. Learn.}, Dec. 2020, pp.
  7--13.

\bibitem{abadi2016deep}
M.~Abadi, A.~Chu, I.~Goodfellow, H.~B. McMahan, I.~Mironov, K.~Talwar, and
  L.~Zhang, ``Deep learning with differential privacy,'' in \emph{Proc. ACM
  SIGSAC conf. comput. commun. secur.}, Oct. 2016, pp. 308--318.

\bibitem{wei2020federated}
K.~Wei, J.~Li, M.~Ding, C.~Ma, H.~H. Yang, F.~Farokhi, S.~Jin, T.~Q.~S. Quek,
  and H.~V. Poor, ``Federated learning with differential privacy: Algorithms
  and performance analysis,'' \emph{{IEEE} Trans. Inf. Forensics Secur.},
  vol.~15, no.~1, pp. 3454--3469, Jun. 2020.

\bibitem{Sun21LDP}
L.~Sun, J.~Qian, and X.~Chen, ``{LDP-FL:} practical private aggregation in
  federated learning with local differential privacy,'' in \emph{Proc. Int.
  Joint Conf. Artif. Intell. (IJCAI), 19-27}, Aug. 2021, pp. 1571--1578.

\bibitem{ghodsi2011dominant}
A.~Ghodsi, M.~Zaharia, B.~Hindman, A.~Konwinski, S.~Shenker, and I.~Stoica,
  ``Dominant resource fairness: Fair allocation of multiple resource types,''
  in \emph{Proc. {USENIX} Symp. Netw. Syst. Des. Implementation (NSDI)}, Mar.
  2011.

\bibitem{li2022dplanner}
W.~Li, L.~Xiang, B.~Guo, Z.~Li, and X.~Wang, ``Dplanner: {A} privacy budgeting
  system for utility,'' \emph{{IEEE} Trans. Inf. Forensics Secur.}, vol.~18,
  no.~1, pp. 1196--1210, Feb. 2023.

\bibitem{tholoniat2022packing}
P.~Tholoniat, K.~Kostopoulou, M.~Chowdhury, A.~Cidon, R.~Geambasu,
  M.~L{\'{e}}cuyer, and J.~Yang, ``Packing privacy budget efficiently,''
  \emph{arXiv}, Jan. 2022.

\bibitem{yuan2022Privacyas}
J.~Yuan, S.~Wang, S.~Wang, Y.~Li, X.~Ma, A.~Zhou, and M.~Xu, ``Privacy as a
  resource in differentially private federated learning,'' in \emph{Proc. IEEE
  Int. Conf. Comp. Commun. (INFOCOM)}, May 2023.

\bibitem{kuchler2023cohere}
N.~K{\"{u}}chler, E.~Opel, H.~Lycklama, A.~Viand, and A.~Hithnawi, ``Cohere:
  Privacy management in large scale systems,'' \emph{arXiv}, Jan. 2023.

\bibitem{luo2021privacy}
T.~Luo, M.~Pan, P.~Tholoniat, A.~Cidon, R.~Geambasu, and M.~L{\'{e}}cuyer,
  ``Privacy budget scheduling,'' in \emph{Proc. {USENIX} Symp. Oper. Syst. Des.
  Implementation (OSDI)}, Aug. 2021, pp. 55--74.

\bibitem{dinh2020federated}
C.~T. Dinh, N.~H. Tran, M.~N.~H. Nguyen, C.~S. Hong, W.~Bao, A.~Y. Zomaya, and
  V.~Gramoli, ``Federated learning over wireless networks: Convergence analysis
  and resource allocation,'' \emph{{IEEE/ACM} Trans. Netw.}, vol.~29, no.~1,
  pp. 398--409, Oct. 2021.

\bibitem{shi2020joint}
W.~Shi, S.~Zhou, Z.~Niu, M.~Jiang, and L.~Geng, ``Joint device scheduling and
  resource allocation for latency constrained wireless federated learning,''
  \emph{{IEEE} Trans. Wirel. Commun.}, vol.~20, no.~1, pp. 453--467, Dec. 2021.

\bibitem{li2019fair}
T.~Li, M.~Sanjabi, A.~Beirami, and V.~Smith, ``Fair resource allocation in
  federated learning,'' in \emph{Proc. Int. Conf. Learn. Representations
  (ICLR)}, Apr. 2020, pp. 1--27.

\bibitem{lim2021decentralized}
W.~Y.~B. Lim, J.~S. Ng, Z.~Xiong, J.~Jin, Y.~Zhang, D.~Niyato, C.~Leung, and
  C.~Miao, ``Decentralized edge intelligence: {A} dynamic resource allocation
  framework for hierarchical federated learning,'' \emph{{IEEE} Trans. Parallel
  Distributed Syst.}, vol.~33, no.~3, pp. 536--550, Sep. 2022.

\bibitem{RENYI}
A.~M. Girgis, D.~Data, and S.~N. Diggavi, ``Renyi differential privacy of the
  subsampled shuffle model in distributed learning,'' in \emph{Proc. Adv.
  Neural Inf. Process. Syst. (NeurIPS)}, Dec. 2021, pp. 29\,181--29\,192.

\bibitem{mironov2017renyi}
I.~Mironov, ``R{\'{e}}nyi differential privacy,'' in \emph{Proc. {IEEE} Comput.
  Secur. Found. Symp. (CSF)}, Aug. 2017, pp. 263--275.

\bibitem{uchida2009information}
M.~Uchida and J.~Kurose, ``An information-theoretic characterization of
  weighted alpha-proportional fairness,'' in \emph{Proc. IEEE Int. Conf. Comp.
  Commun. (INFOCOM)}, Apr. 2009, pp. 1053--1061.

\bibitem{lan2010axiomatic}
T.~Lan, D.~T.~H. Kao, M.~Chiang, and A.~Sabharwal, ``An axiomatic theory of
  fairness in network resource allocation,'' in \emph{Proc. IEEE Int. Conf.
  Comp. Commun. (INFOCOM)}, Mar. 2010, pp. 1343--1351.

\bibitem{joe2013multiresource}
C.~Joe-Wong, S.~Sen, T.~Lan, and M.~Chiang, ``Multiresource allocation:
  Fairness--efficiency tradeoffs in a unifying framework,'' \emph{IEEE/ACM
  Trans. Netw.}, vol.~21, no.~6, pp. 1785--1798, May 2013.

\bibitem{shokri2017membership}
R.~Shokri, M.~Stronati, C.~Song, and V.~Shmatikov, ``Membership inference
  attacks against machine learning models,'' in \emph{Proc. IEEE Symp. Secur.
  Privacy (S\&P)}, Jun. 2017, pp. 3--18.

\bibitem{carlini2019secret}
N.~Carlini, C.~Liu, {\'U}.~Erlingsson, J.~Kos, and D.~Song, ``The secret
  sharer: Evaluating and testing unintended memorization in neural networks,''
  in \emph{Proc. USENIX Secur. Symp. (USENIX Security)}, Aug. 2019, pp.
  267--284.

\bibitem{lai2021oort}
F.~Lai, X.~Zhu, H.~V. Madhyastha, and M.~Chowdhury, ``Oort: Efficient federated
  learning via guided participant selection,'' in \emph{Proc. {USENIX} Symp.
  Oper. Syst. Des. Imple., (OSDI)}, Jul. 2021, pp. 19--35.

\bibitem{Generalized2008Eitan}
E.~Altman, K.~Avrachenkov, and A.~Garnaev, ``Generalized a-fair resource
  allocation in wireless networks,'' in \emph{Proc. Conf. Decis. Control
  (CDC)}, Dec. 2008, pp. 2414--2419.

\bibitem{gurobi2022gurobi}
\BIBentryALTinterwordspacing
{Gurobi Optimization, LLC}, ``Gurobi optimizer reference manual,'' 2022.
  [Online]. Available:
  \url{https://www.gurobi.com/documentation/current/refman/index.html}
\BIBentrySTDinterwordspacing

\bibitem{tang2006counter}
A.~Tang, J.~Wang, and S.~H. Low, ``Counter-intuitive throughput behaviors in
  networks under end-to-end control,'' \emph{{IEEE/ACM} Trans. Netw.}, vol.~14,
  no.~2, pp. 355--368, Jul. 2006.

\bibitem{renyi1961measures}
A.~R{\'e}nyi, ``On measures of entropy and information,'' in \emph{Proc.
  Berkeley Symp. Math. Statist. Probability, Volume 1: Contributions to the
  Theory of Statistics}, vol.~4, 1961, pp. 547--562.

\bibitem{menezes1970theory}
C.~F. Menezes and D.~L. Hanson, ``On the theory of risk aversion,'' \emph{Int.
  Econ. Rev.}, vol.~11, no.~3, pp. 481--487, 1970.

\end{thebibliography}
